\documentclass[10pt,draft,reqno]{amsart}
\usepackage{amsmath,latexsym,amssymb,ifpdf}
\usepackage{color,graphicx}

\graphicspath{{figures/}}

     \makeatletter
     \def\section{\@startsection{section}{1}%
     \z@{.7\linespacing\@plus\linespacing}{.5\linespacing}%
     {\bfseries
     \centering
     }}
     \def\@secnumfont{\bfseries}
     \makeatother
\newtheorem{theorem}{Theorem}[section]
\newtheorem{lemma}[theorem]{Lemma}
\newtheorem{proposition}[theorem]{Proposition}
\newtheorem{corollary}[theorem]{Corollary}
\theoremstyle{definition}
\newtheorem{definition}[theorem]{Definition}

\newtheorem{assumption}[theorem]{Assumption}
\theoremstyle{remark}
\newtheorem{remark}[theorem]{Remark}
\numberwithin{equation}{section} 

\renewcommand{\paragraph}[1]{{\bf #1.}}
\definecolor{myred}{rgb}{0.8,0,0}  

{\vskip\baselineskip\noindent\textbf{Proof of {#1}:}}%
{\hspace*{.1pt}\hspace*{\fill}\BOX\vskip\baselineskip}

\def \R{\mathbb{R}}               
\def \N{\mathbb{N}}               
\def \1{{\bf 1}}                
\def \0{{\bf 0}}
\def \tr{^{\!\top}}             

\def \pini{\widetilde{p}}
\def \prini{\widetilde{\pr}}
\def \simplex{\mathcal{S}}
\def\eps{\varepsilon}
\def\epsmax{\overline{\varepsilon}}
\def\compactset{K}
\def\gammaI{{\gamma_I}}
\def\gammaN{{\gamma}}
\def\np{m}

\def \R{\mathbb{R}}               
\def \N{\mathbb{N}}               
\def \1{{\bf 1}}                
\def \0{{\bf 0}}
\def \tr{^{\!\top}}             

\def \pini{\widetilde{p}}
\def \pr{\pi}
\def \simplex{\mathcal{S}}
\def\eps{\varepsilon}

\def \piq{{\overline \pi}}
\def \tq{{\overline t}}

\def\epsmax{\overline{\varepsilon}}
\def\compactset{K}
\def\gammaI{\gamma_I}
\def\gammaN{\gamma}
\def\np{m}
\def \dist{\text{\rm dist}\,}

\newcommand{\norm}[1]{\left|{#1}\right|}
\newcommand{\normmax}[1]{\norm{#1}_\infty}
\newcommand{\restr}[1]{\underline{#1}}
\def \simplexr{\restr{\simplex}}
\def \alphar{\restr{\alpha}}
\def \betar{\restr{\beta}}
\def \gammar{\restr{\gamma}}
\def \gammaIr{\restr{\gamma}_I}

\begin{document}

\title[Portfolio Optimization with Expert Opinions]{ Portfolio
Optimization under Partial Information with Expert Opinions: a Dynamic Programming Approach}

\author{R\"udiger Frey}
\address{R\"udiger Frey, Institute for Statistics and Mathematics,
 Vienna University of Economics and Business,
 Augasse 2-6, A-1090 Vienna, Austria}
\email{ruediger.frey@wu.ac.at}

\author[Abdelali Gabih]{Abdelali Gabih}
\address{Abdelali Gabih, Universit\'{e} Cadi Ayyad, ENSA Marrakech, Laboratoire OSCARS, Boulevard
Abdelkarim Khattabi Gu\'{e}liz BP 575, 40000 Marrakech, Morocco}
\email{a.gabih@uca.ma}

\author[Ralf Wunderlich]{Ralf Wunderlich}
\address{Ralf Wunderlich, Mathematical Institute, Brandenburg
University of Technology Cottbus -- Senftenberg ,  Postfach 101344, D-03013 Cottbus,
Germany} \email{ralf.wunderlich@tu-cottbus.de }

\date{ \today \quad This version of the paper grew out of an earlier, unpublished version that was authored only by Frey and Wunderlich; see \cite{Frey-Wunderlich-2013}.}



\begin{abstract}
This paper investigates optimal portfolio strategies in a market
where the drift is driven by an unobserved Markov chain. Information
on the state of this chain is obtained from  stock prices and
expert opinions in the form of signals at random discrete time
points. As in Frey et al.~(2012), Int.~J.~Theor.~Appl.~Finance, 15,
No.~1, we use stochastic filtering to transform the original problem
into an optimization problem under full information where the state
variable is the filter for the Markov chain.  The dynamic programming equation for this problem is studied with viscosity-solution  techniques and with  regularization
arguments.
\end{abstract}

\maketitle

\section{Introduction}
It is well-known that optimal investment strategies in dynamic
portfolio optimization depend crucially on the drift of the
underlying asset price process. On the other hand it is notoriously
difficult to estimate drift parameters from historical asset price
data. Hence it is natural to include  expert opinions or investors'
views as additional source of information in the computation of
optimal portfolios. In the context of the classical one-period
Markowitz model this leads to the well-known Black-Littermann
approach,  where Bayesian updating is used to improve return
predictions  (see Black \& Litterman
\cite{Black_Litterman (1992)}).

Frey et al.~\cite{Frey et al. (2012)} consider expert opinions in
the context of a dynamic portfolio optimization problem in
continuous time. In their paper the asset price process is modelled
as diffusion whose drift is driven by a hidden finite-state Markov
chain $Y$. Investors observe the  stock prices and in addition a
marked point process with  jump-size distribution depending on the
current state of $Y$ that represents expert opinions. Frey et
al.~\cite{Frey et al. (2012)}  derive a finite-dimensional filter
$p_t$ with jump-diffusion dynamics for the state of $Y$ and they  reduce the portfolio
optimization problem to a problem under complete information with
state variable given by the filter $p_t$. Moreover they write down
the dynamic programming equation for the value function $V$ of that
problem and, assuming that the dynamic programming equation admits a classical solution,  they compute a candidate solution for the optimal strategy.  The precise mathematical
meaning of these preliminary results is however left open.

This issue is addressed in the present paper. A major challenge
in the analysis of the dynamic programming equation is the fact that the
equation is not strictly elliptic if
the number of states of $Y$ is larger than the number of assets. In fact, due to this non-ellipticity
it is not possible to apply any of the  known results on the existence of classical solutions to this equation.
We study two ways to address this problem. First, following the analysis
of Pham \cite{Pham (1998)} we show that  the value function is a  viscosity solution of the associated dynamic programming equation. Since the comparison principle for viscosity solutions applies  to our model, this yields an elegant characterization of the value function. However, the viscosity-solution methodology  does not provide any information on the form of (nearly) optimal strategies.

For this reason we study a second approach based on regularization arguments. Here  an additional noise term of the form $m^{- \frac{1}{2}}
d\widetilde{B}_t$, $ \widetilde{B} $ an independent  Brownian motion
of suitable dimension and $m\in\N $ large, is added to the dynamics
of the state process $p$. The dynamic programming equation
associated with the regularized optimization problem is strictly
elliptical so that  recent results of Davis \& Lleo
\cite{Davis and Lleo (2012)} imply  the existence of a classical solution $V^m$.
Moreover, the optimal strategy for the regularized problem can be characterized as solution of a quadratic optimization problem that involves $V^m$ and its first derivatives. We show that for $m \to \infty $ reward- and value function for the regularized problem and
the original problem converge uniformly for all admissible strategies.  This uniform convergence  implies that for $m$ sufficiently large the  optimal strategy for the regularized problem is a nearly-optimal strategy in the original problem, so that we have
solved the problem of finding good strategies.
In order to carry out this program we need an explicit representation of jump-diffusion processes as a solution of an SDE driven by Brownian motion and - this is the new part - some \emph{exogenous} Poisson random measure; we  refer the reader to Section \ref{spec_example} below for details.

The related literature on portfolio optimization under partial
information is discussed in detail in the companion paper \cite{Frey
et al. (2012)}. Here we just mention the papers Rieder \& B\"aeuerle
\cite{Rieder_Baeuerle2005} and Sass \& Haussmann\cite{Sass and
Haussmann (2004)} that are concerned with portfolio optimization in
models with Markov-modulated drift but without any extra
information.

The paper is organized as follows. In Section \ref{market_model} we
introduce  the model of the financial market and formulate the
portfolio optimization problem. For this problem we derive in Section
\ref{power_utility}  the dynamic programming equation  in the case
of power utility. In Section \ref{reformulation} we reformulate the state
equation in terms of an exogenous Poisson random measure. For this
reformulated state equation we provide in Section \ref{spec_example}  an
explicit construction of the jump coefficient.
The main results of this paper are presented in Sections \ref{viscosity_sol} and  \ref{regularization}. Here we  show that the value function is a viscosity solution of the dynamic programming equation. Moreover, we study
a regularized version of the dynamic programming equation and
investigate  nearly optimal strategies.

\section{Model and optimization problem}

\label{market_model} The  setting is based on
\cite{Frey et al. (2012)}. For a  fixed date $T>0$
representing the investment horizon, we work on a filtered
probability space $(\Omega,\mathcal{G},\mathbb{G},P)$, with
filtration $\mathbb{G}=(\mathcal {G}_t)_{t \in [0,T]}$ satisfying
the usual conditions. All processes are assumed to be
$\mathbb{G}$-adapted. For a generic $\mathbb{G}$-adapted process $H$
we denote by $\mathbb{G}^H$ the filtration generated by  $H$.

\paragraph{Price dynamics}
 We consider a market model  for one risk-free bond with price
$S^0_t=1$ and  $n$ risky securities with prices
$S_t=(S_t^1,\ldots,S^n_t)\tr$   given by
\begin{eqnarray}
\label{stockmodel}
dS^i_t=S^i_t\,\Big(\mu^i(Y_t)dt+\sum_{j=1}^n\sigma^{ij}dW^j_t\Big),
\quad S^i_0=s^i, \quad i=1,\cdots,n.
\end{eqnarray}
Here $\mu=\mu(Y_t) \in \R^n$ denotes the mean stock return or drift
which is driven by some factor process $Y$ described below. The
volatility $\sigma=(\sigma^{ij})_{1\leq i,j\leq n}$ is assumed to be
a constant invertible matrix and $W_t=(W_t^1,\cdots.W_t^n)$ is an
$n$-dimensional $\mathbb{G}$-adapted Brownian motion.  The
invertibility of $\sigma$ always can be ensured by a suitable
parametrization if the  covariance matrix $\sigma\sigma\tr$ is
positive definite.  The factor process $Y$  is a finite-state Markov
chain independent of the Brownian motion $W$ with state space
$\{e_1,\ldots,e_d\}$ where $e_i$ is the $i$th unit vector in $\R^d$.
The generator matrix is denoted by $Q$ and the  initial distribution
by $\pini=(\pini^1,\ldots,\pini^d)\tr $. The states of the factor
process $Y$ are mapped onto the states $\mu_1,\ldots,\mu_d$ of the
drift by the function $\mu(Y_t)=MY_t$, where $M_{lk} =
\mu_k^l=\mu^l(e_k)$, $1 \le l \le n, \, 1 \le k \le d$.

Define the return process $R$  associated with the price process $S$
by $dR_t^i=dS_t^i/S_t^i$, $i=1,\ldots,n $.  Note that $R$ satisfies
$dR_t = \mu(Y_t)dt + \sigma dW_t, $ and it is easily seen that
$\mathbb{G}^R = \mathbb{G}^{\log S} = \mathbb{G}^S\,.$ This is
useful, since it allows us to work with $R$ instead of $S$ in the
filtering part. For details we refer to \cite{Frey et al. (2012)}.

\paragraph{Investor Information}
We assume that the investor does not observe the factor process $Y$
directly; he does however know the model parameters, in particular
the initial distribution $\pini$, the generator matrix $Q$ and the
functions $\mu^i(\cdot)$. Moreover, he  has  noisy observations of
the hidden process  $Y$ at his disposal. More precisely we assume
that the investor  observes the return process $R$ and that he
receives at discrete points in time $T_n$ noisy signals about the
current state of $Y$. These signals are to be interpreted as expert
opinions; specific examples can be found in the companion paper
\cite{Frey et al. (2012)}.

We model expert opinions  by a marked point process $I=(T_n,Z_n)$,
so that at $T_n$ the investor observes the realization of a random
variables $Z_n$ whose distribution depends on  the current state
$Y_{T_n}$ of the factor process. The $T_n$ are modeled as jump times
of a standard Poisson process with intensity $\lambda$, independent
of $Y$, so that the timing of the information arrival does not carry
any useful information. The signal $Z_n$ takes values in some set
$\mathcal{Z}\subset \R^\kappa$, and we assume that given $Y_{T_n} =
e_k$, the distribution of $Z_n$  is absolutely continuous with
Lebesgue-density $f_k(z)$. We identify the marked point process
$I=(T_n,Z_n)$ with the associated counting measure denoted by $I(dt,
dz)$. Note that the $\mathbb{G}$-compensator of $I$ is $\lambda dt
\sum_{k=1}^d 1_{\{Y_t = e_k\}} f_k(z)dz$.

Summarizing, the information available to the investor is given by
the \emph{investor filtration }  $\mathbb{F}$ with
\begin{equation}
\mathcal{F}_t = \mathcal{G}_t^R \vee \mathcal{G}_t^I \,,\quad 0 \le
t \le T.
\end{equation}

\vspace*{0ex}
\paragraph{Portfolio and optimization problem}
We describe the selffinancing trading of an investor by the initial
capital $x_0>0$ and the $n$-dimensional $\mathbb{F}$-adapted trading
strategy $h $ where $h_t^i$, $i=1,\ldots,n$, represents  the
proportion of wealth  invested in stock $i$ at time $t$.  It is
well-known that in this setting the wealth process $X^{(h)}$ has the
dynamics
\begin{eqnarray} \label{wealth_phys}
\frac{dX_t^{(h)}}{X_t^{(h)}}=\sum_{i=0}^n h^i_t\frac{dS^i_t}{S^i_t}
&= &h_t^{\top}\mu(Y_t) dt+h_t^{\top}\sigma dW_t,\quad X_0^{(h)}=x_0.
\end{eqnarray}
We assume that for all $t\in[0,T]$ the strategy $h_t$ takes values
in some non-empty convex and  compact subset $\compactset$ of $\mathbb{R}^n$ that can be described in terms of $r$ linear constraints.  In mathematical terms,
\begin{equation} \label{eq:def-K}
 \compactset = \{ h \in \R^n \colon \Psi_l^\top h \le \nu_l,\,  1 \le l \le  r,\, \text{ for given } (\Psi_1, \nu_1), \dots , (\Psi_r, \nu_r) \in \R^n \times \R \}\,.
\end{equation}
We assume that there is some $h^0 \in \R^n$ such that $\Psi_l^\top h^0 < \nu_l$ for all $1 \le l \le r$ and that  $0 \in \compactset$.
The set $\compactset$  models constraints on the portfolio.  Moreover, the assumption that $h_t \in \compactset$ for
all $t$ facilitates many technical estimates in the paper. For  a specific example fix constants  $ c_1 <0$, $c_2 >1$, and let
$$
\compactset = \{ h \in \R^n\colon h_i \ge c_1 \text{ for  all } 1 \le i \le n\,
\text{ and }\sum_{i=1}^n h_i \le c_2  \}\,.
$$
This choice of $\compactset$ hat would correspond to a limit $|c_1|$ on the amount of  shortselling and a limit $c_2$ for leverage.

We denote the  class of {\em admissible trading strategies} by
\begin{equation}
\label{set_admiss} \mathcal{H}=\{h= (h_t)_{t\in[0,T]}  \colon \text{
$h$ is $\mathbb{F}$ adapted and } h_t\in \compactset \text{ for all
} t \}\,.
\end{equation}
Since $\mu(Y_t)$ is bounded and since    $\sigma$ is constant,  equation
(\ref{wealth_phys}) is well defined for all $h\in\mathcal{H}$.

We assume that the investor wants to maximize the expected utility
of terminal wealth for power utility $U(x)=\frac{x^\theta}{\theta}$,
$\theta<1,\;\theta\not=0$.\footnote{ The case $\theta =0$
corresponds to logarithmic utility $U(x) = \ln x$ which is treated
in \cite{Frey et al. (2012)}.} The optimization problem thus reads
as
\begin{eqnarray}
\label{opti_org}
\max \{ E(U(X_T^{(h)})) \colon {h\in\mathcal{H}}\}.
\end{eqnarray}
This is a maximization problem under partial information since  we
have required that the strategy $h$ is adapted to the investor
filtration  $\mathbb{F}$.

\paragraph{Partial information and filtering}
Next we explain how the control problem (\ref{opti_org}) can be
reduced to a control problem with complete information via filtering
arguments. We  use the following notation: for a generic process $H$
we denote by $\widehat{H}_t=E(H|\mathcal{F}_t)$ its optional
projection on  the filtration $\mathbb{F}$, and   the filter for the
Markov chain $Y_t$ is denoted by $p_t=(p_t^1,\cdots,p_t^d)$ with
$p^k_t=P(Y_t=e_k|\mathcal{F}_t),\; k=1\ldots,d$.  Note that for a
process of the form $H_t=h(Y_t)$ the optional projection is given by
$\widehat{h(Y_t)}=\sum_{k=1}^d h(e_k)p_t^k$. In particular, the
projection of of the drift equals
$$ \widehat{\mu(Y_t)} = \sum_{k=1}^d \mu(e_k) p_t^k = M p_t\,.$$
The following two processes will drive the dynamics of $p_t$. First,
let
\begin{eqnarray*}
{\widetilde W}_t&:=&\sigma^{-1}(R_t - \int_0^t M p_s ds).
\end{eqnarray*}
By standard results from filtering theory $\widetilde W$ is an
$\mathbb{F}$-Brownian motion (the so-called innovations process).
Second, define the predictable random measure
 \[\nu_I(dt, dz)=\lambda  dt\sum_{k=1}^{d}p^k_{t-}f_k(z)dz.\]
By standard results on point processes $\nu_I$ is  the
$\mathbb{F}$-compensator of $I$, see for instance Bremaud
\cite{Bremaud}. The compensated random measure will be denoted by
${\widetilde{I}}(dt, dz) := I(dt, dz) -\nu_I(dt, dz)$.

Using a combination of the HMM filter (see e.g.~Wonham \cite{Wonham
(1965)}, Elliott et al.~\cite{Elliott et al. (1994)}, Liptser \&
Shiryaev \cite{Liptser-Shiryaev}) and Bayesian updating, in
 \cite{Frey et al. (2012)}   the following
$d$-dimensional SDE system for the dynamics of the filter $p$ is
derived
\begin{equation}
\begin{split}
 dp_t&=Q\tr p_t dt+ \beta\tr(p_t) d{\widetilde W}_t +
 \int_{\mathcal{Z}} \gammaI(p_{t-},z) \widetilde{I}(dt, dz)
\label{filter_update}
\end{split}
\end{equation}
with initial condition  $p_0^k = \pini^k.$ Here, the matrix  $\beta
=\beta(p) = (\beta_1,\ldots,\beta_d) \in \R^{n\times d}$ and the
vector $\gammaI  =\gammaI(p,z) = (\gammaI^1,\ldots,\gammaI^d)\tr \in
\R^{ d}$  are defined by
\begin{equation}
\label{beta_gamma_def}
\begin{split}
\beta_k(p)&= p^k\Big(\sigma^{-1}\big(\mu_k-\sum_{j=1}^{d}p^j\mu_j\big)  = p^k\sigma^{-1}M(e_k-p) \in \R^n\\
\quad\text{and}\quad \gammaI^k(p,z)&=
p^k\bigg(\frac{f_k(z)}{\overline f(z,p)}-1\bigg),\, 1 \le k \le d,
\quad\text{with}\quad \overline{f}(z,p) =  \sum_{k=1}^d p^k f_k(z).
\end{split}
\end{equation}

It is well-known (see e.g. Lakner \cite{Lakner (1998)}, Sass \&
Haussmann \cite{Sass and Haussmann (2004)}) that the
$\mathbb{F}$-semimartingale decomposition of  $X$ is given by
\begin{eqnarray}
\label{V_semimart} \frac{dX_t^{(h)}}{X_t^{(h)}} &=&h_t^{\top}\, M
p_t \,dt+h_t^{\top}\sigma d{\widetilde W}_t.
\end{eqnarray}
Now note that for  a constant strategy  $h_t\equiv h\in\compactset$
the $(d+1)$-dimensional process $(X^{(h)},p)$ is an
$\mathbb{F}$-Markov process as is immediate from the dynamics in
\eqref{filter_update} and \eqref{V_semimart}. Hence  the
optimization problem (\ref{opti_org}) can be considered as a control
problem under complete information with the $(d+1)$-dimensional
state variable process $(X^{(h)},p)$. This control problem is
studied in the remainder of the paper.

\section{Dynamic programming  equation for the case of power utility}
\label{power_utility}
\vspace{0.1cm}
\paragraph{A simplified optimization problem}
As a first step,  we simplify the control problem by a change of
measure. As shown in Nagai \& Runggaldier \cite{Nagai and
Runggaldier (2008)} this measure change leads to a new problem where
the set of state variables is reduced to $p$ and where the dynamic
programming equation takes on a simpler form. First we compute for
an admissible strategy $h\in \mathcal{H}$ the utility of terminal
wealth $U(X_T^{(h)}) = \frac{1}{\theta}(X_T^{(h)})^\theta$. From
(\ref{V_semimart}) it follows that
\begin{equation}\label{eq:power-wealth}
\frac{1}{\theta}(X_T^{(h)})^\theta = \frac{x_0^{\theta}}{\theta}
\exp\Big\{\theta\int_0^T\Big(h_s^{\top}\,Mp_s -
\frac{1}{2} \norm{\sigma^\top h_s}^2\Big)ds + \theta\int_0^T
h_s\tr\sigma d{\widetilde W}_s\Big\},
\end{equation}
where $\norm{.}$ denotes the Euclidean norm.
Define now the random variable
 $L_T^{(h)} = \exp\big\{\int_0^T \theta
h_s^{\top}\sigma d{\widetilde W}_s -
            \frac{1}{2}\int_0^T \norm{\theta \sigma^{\top}h_s}^2  ds\big\}$ and the function
\begin{equation}
 \label{eq:def-b}
  b(p,h;\theta) = -\theta\Big(h^{\top} Mp - \frac{1-\theta}{2} \norm{\sigma^{\top} h }^2\Big).\\
\end{equation}
With this notation  \eqref{eq:power-wealth} can be written in the form
\begin{equation}
\frac{1}{\theta}(X_T^{(h)})^\theta = \frac{x_0^{\theta}}{\theta}
\,L_T^{(h)}\, \exp\Big\{\int_0^T - b(p_s,h_s; \theta) ds\Big\}\,.
\end{equation}

Since $\sigma $ is deterministic  and since $h$ is bounded,  the Novikov
condition  implies that  $E(L_T^{(h)})=1$. Hence we can define an
equivalent measure $P^h$ on $\mathcal{F}_T$  by $dP^h /dP =
L_T^{(h)}$, and Girsanov's theorem guarantees that $B_t:={\widetilde
W}_t-\theta \int_0^t\sigma^{\top} h_s ds$  is a standard
$\mathbb{F}$-Brownian motion.  Substituting into
(\ref{filter_update}) we find the following dynamics for the filter
under $P^h $
\begin{eqnarray}
\label{filter_h}
 dp_t&=&\alpha(p_t,h_t) dt+ \beta\tr(p_t) dB_t +
 \int_{\mathcal{Z}} \gammaI(p_{t-},z) \widetilde{I}(dt, dz)
\\
\label{alpha_def} \text{where}\quad \alpha&=&\alpha(p,h)= Q^{\top} p
+ \theta \beta^{\top}(p)\sigma^{\top} h .
\end{eqnarray}
In view of these transformations, for $0< \theta <1$ the
optimization problem (\ref{opti_org}) is equivalent to
\begin{equation} \label{eq:objective-1}
\max \Big \{  E \Big(\exp\Big\{\int_0^T
-b(p_s^{(0,\pini,h)},h_s;\theta) ds \Big\}\Big) \colon h \in
\mathcal{H} \Big \}
\end{equation}
 where  we denote by $p_s^{(t,p,h)}$  the solution of (\ref{filter_h})
 for $s\in [t,T]$ starting at time $t\in[0,T]$  with initial value $p\in \simplex$
 for strategy $h\in \mathcal{H}$.
For $\theta < 0$ on the other hand (\ref{opti_org}) is equivalent to minimizing
the expectation in \eqref{eq:objective-1}. In the sequel we will
concentrate on the case  $0< \theta <1$; the necessary changes for
$\theta <0$ will be indicated where appropriate. Moreover, $\theta$
will be largely removed from the notation.  The reward and value
function  for this control problem are given by
\begin{eqnarray}
\nonumber v(t,p,h) &=& E \Big( \exp\Big\{\int_t^T -b
(p_s^{(t,p,h)},h_s) ds\Big\}\Big)
\quad \text{for } h\in \mathcal{H},\\
\label{reward_value_function}\\[-2ex]
\nonumber V(t,p) & = &\sup \{ v(t,p,h)\colon h \in \mathcal{H}\}.
\end{eqnarray}
Note that $v(T,p,h) = V(T,p)=1$.

\paragraph{The dynamic programming equation} Next, we derive  the form of the dynamic programming equation for
$V(t,p)$. We begin with the generator of the \textit{state process}
$p_t$  (the solution of the SDE (\ref{filter_h})) for a constant strategy $h_t \equiv h$.  Denote by   $\mathcal{S}=\{p\in \R^d:
\sum_{i=1}^d p^i =1, p^i\ge 0, i=1,\ldots,d\}$  the unit simplex in
$\R^d$. Standard arguments show that the solution of this SDE is a Markov process
whose generator $\mathcal{L}^{h}$ operates on $g\in \mathcal{C}^2
(\mathcal{S})$ as follows
\begin{eqnarray}\label{new-generator}
\mathcal{L}^{h}g(p)=
\frac{1}{2}\sum_{i,j=1}^{d}\beta_i^{\top}(p)\beta_j(p)g_{p^ip^j}
&+& \sum_{i=1}^{d} \alpha^i(p,h)g_{p^i}
\\ \nonumber
&+&
\lambda\int_{\mathcal{Z}}\{g(p+\gammaI(p,z))-g(p)\}\overline{f}(z,p)dz.
\end{eqnarray}

By standard arguments the dynamic programming equation associated to
this optimization problem is
\begin{eqnarray}\label{bell-equa0}
V_t(t,p)+\sup_{h \in \compactset}
\Big\{\mathcal{L}^{h}V(t,p)-b(p,h;\theta)V(t,p)\Big\}=0, \quad (t,p)
\in [0,T) \times \mathcal{S},
\end{eqnarray}
with terminal condition  $V(T,p)=1.$ In case that $\theta <0$ the
equation is similar, but the $\sup$ is replaced by an $\inf$.
Plugging in $\mathcal{L}^{h}$ as given in (\ref{new-generator}) and
$b(p,h)$ as given in \eqref{eq:def-b} into ($\ref{bell-equa0}$) the
dynamic programming  equation can be written more explicitly as
\begin{align}
\nonumber
0&=V_t(t,p)+
\frac{1}{2}\sum_{k,l=1}^{d}\beta_k^{\top}(p_t)\beta_l(p_t)V_{p^kp^l}(t,p) +
\sum_{k=1}^{d}\Big\{\sum_{l=1}^{d}Q^{lk}p^l\Big\}V_{p^k}(t,p) \\
\label{Bell-equa1}
&+ \lambda\int_{\mathcal{Z}}\{V(t,p+\gammaI(p,z))-V(t,p)\}\overline{f}(z,p)dz\\
&+
\sup_{h\in\compactset}\Big\{\sum_{k=1}^{d}\beta_k^{\top}(p_t)\sigma^{\top}\theta
hV_{p^k}(t,p)+\theta V(t,p)\Big(h^{\top} M p-\frac{1}{2}
\norm{\sigma^{\top}h}^2 (1-\theta)\Big)\Big\}.
\nonumber
\end{align}
Suppose for the moment that a classical solution to
(\ref{Bell-equa1}) exists.  The argument of the supremum
in the last line of \eqref{Bell-equa1} is quadratic in $h$ and strictly concave (as $\sigma \sigma^\top$ is positive definite). Hence this function attains a unique maximum $h^*$ on the convex set $\compactset$. Moreover, as shown in Davis and Lleo \cite{Davis and Lleo (2012)}, Proposition~3.6,  $h^*$ can be chosen as a measurable function of $t$  and $p$.   Hence there exists a solution $p^*$ of the SDE
\eqref{filter_h} with $h_t = h^*(t,p_t^*)$; this can be verified by a similar application of the Girsanov theorem
as in the derivation of the equation \eqref{filter_h}.
Then standard verification arguments along the lines of  Theorem~3.1 of Fleming \& Soner \cite{bib:fleming-soner-06} or Theorem~5.5 of Davis and Lleo \cite{Davis and Lleo (2012)}
immediately give that $V$ is the value function of the control problem
\eqref{eq:objective-1} and that $h_t^* := h^*(t,p_t^*)$ is the optimal strategy.

\begin{remark}\label{rem:strategy}
If  for some $(t,p)$ $h^*(t,p)$ is inner point of $\compactset$, an explicit formula for $h^*(t,p)$ can be given. In that case $h^*(t,p)$ is  given by the solution $h^*$ of the following linear equation
(the first-order condition for the unconstrained problem)
\begin{eqnarray*}
\sigma\sum_{k=1}^{d}\beta_k(p) V_{p^k}(t,p)+ V(t,p)\Big( M p
-\sigma\sigma^{\top}h(1-\theta)\Big)=0.
\end{eqnarray*}
Since $\sigma$ is an invertible matrix $h^*$ equals
\begin{equation*}
h^*=h^*(t,p) = \frac{1}{(1-\theta)}(\sigma\sigma^{\top})^{-1}\Big\{M
p+\frac{1}{V(t,p)}\sigma \sum_{k=1}^d \beta_k(p)V_{p^k}(t,p) \Big\}.
\end{equation*}
\end{remark}

However, the existence of a classical solution of equation~\eqref{Bell-equa1} is an open
issue. The main problem is the fact that one cannot guarantee that  the  equation is
\emph{uniformly elliptic}. To see this note that the coefficient matrix  of the second
derivatives in~\eqref{Bell-equa1} is  given by  $C(p) =  \beta^{\top}(p) \beta(p) $.
By definition equation~\eqref{Bell-equa1} is uniformly elliptic if
the matrix $C(p)$ is strictly positive
definite uniformly in $p$. A necessary condition for this is that  there are no non-trivial solutions of the linear equation $\beta x =0$ so that we need to have the inequality $n\ge d$ (at least as many assets as states of the Markov chain $Y$). Such an assumption is hard to
justify economically; imposing it nonetheless out  of mathematical necessity
would severely limit the applicability of our approach.

In the  present paper we therefore  study two alternative routes to giving a precise mathematical meaning to the dynamic programming equation \eqref{Bell-equa1}. First, following the analysis
of Pham \cite{Pham (1998)}, in Section~\ref{viscosity_sol} we show that  the value function is a viscosity solution of the associated dynamic programming equation. Since the comparison principle for viscosity solutions applies in our case, this provides an elegant characterization of the value function. However, the viscosity-solution methodology  does not provide any information on the form of the optimal strategies. For this reason,
in Section~\ref{regularization}
we use regularization arguments to find  approximately optimal
strategies. More precisely, we add  a term $\frac{1}{\sqrt{\np}}  d \widetilde{B}_t$,
with $\np\in \N$ and $\widetilde B$ a  Brownian
motion of suitable dimension and independent of $B$, to the dynamics of the state equation
(\ref{filter_h}). The HJB equation associated with these regularized
dynamics has an additional term $\frac{1}{2\np} \Delta V$, $\Delta$ the Laplace operator,  and is therefore uniformly elliptic. Hence the results of
Davis \& Lleo \cite{Davis and Lleo (2012)} apply directly to the modified
equation, yielding the existence of a classical solution
$V^\np $. Moreover, the optimal strategy ${}^\np h^*$
of the  regularized problem is  given by the argument of the supremum
in the last line of \eqref{Bell-equa1}   with
$V^\np $ instead of $V$. We then derive convergence results for the reward- and the value function of the regularized problem as $m \to \infty$. In particular, we show in Theorem~\ref{eps_opt} that for $\np$ sufficiently large ${}^\np h^*$ is
approximately optimal in the original problem.

\section{Reformulation of the State Equation}
\label{reformulation} To carry out the program described above we
have to reformulate the state equation for a number of reasons.
First, in our model the state variable process $p$  (the solution of (\ref{filter_h}) takes values in the
simplex $\simplex$ which is a subset of a $d-1$-dimensional
hyperplane  of $\R^d$. If we introduce the announced regularization
to the diffusion part of the state equation then the state variable
will leave this hyperplane  and takes values in the whole $\R^d$ so
that the normalization property of $p$ is violated, which creates
technical difficulties. Second,  in our analysis  we need to apply
results from the literature on the theory of dynamic programming of
controlled jump diffusions, such as Pham \cite{Pham (1998)} and
Davis \& Lleo \cite{Davis and Lleo (2012)}. These papers consider
models where the jump part of the state variable is driven by an
\emph{exogenous} Poisson random measure, and this structure is in fact essential for many arguments in these papers.  In our model, on the other hand,
the measure $\widetilde{I}$ is not an exogenous Poisson random
measure since the law of the compensator $\nu_I$ depends on the
solution $\pi_t$. Hence we need to reformulate the dynamics of the state variable process in terms of an exogenous Poisson random measure.

\paragraph{Restriction to a $d-1$-dimensional state}
We rewrite the state equation in terms of the `restricted'
($d-1$)-dimensional process
$\pr = (\pr^1,\ldots,\pr^{d-1})\tr = (p^1,\ldots,p^{d-1})\tr.$
Then the original state $p$ can be recovered from $\pr$ by using the
normalization property for the last component $p^d$ and we define
$p=R \pr := \big(\pr_1,\dots,\pr_{d-1}, 1 -\sum_{i=1}^{d-1} \pr^i
\big)^\top$. Assuming  $p\in \simplex $ implies that the restricted
state process takes values in
$$
\simplexr =\Big\{\pr\in \R^{d-1}:  \sum_{i=1}^{d-1} \pr^i \le 1,
\pr^i\ge 0, ~i=1,\ldots,d-1 \Big\}.
$$
Now the state equation for $\pr\in \simplexr$ associated to
\eqref{filter_h} reads as
\begin{eqnarray}
\label{filter_restr_I}
 d\pr_t&=&\alphar(\pr_t,h_t) dt+ \betar\tr(\pr_t) dB_t +
 \int_{\mathcal{Z}} \gammaIr(\pr_{t-},z) \widetilde{I}(dt, dz)
\end{eqnarray}
where the coefficients are given by
\begin{align}
\restr{\alpha}(\pr,h) &=({\alpha}^1( R \pr, h),\ldots, {\alpha}^{d-1}( R \pr, h))\tr\in \R^{d-1}\\
\restr{\beta}(\pr) &= ({\beta}_1(R \pr),\ldots, {\beta}_{d-1}(
R\pr))\in
\R^{n\times d-1} \\
\restr{\gamma}_I(\pr,z )& =({\gamma}^1_I( R\pr,z),\ldots,
{\gamma}^{d-1}_I( R\pr,z) )\tr\in \R^{d-1}.
\end{align}
It is straightforward  to give an explicit expression for
$\restr{\alpha}$, $\restr{\beta}$ and $\restr{\gamma}_I$, but such
an expression  is not needed in the sequel. The original state  can
be recovered from $\pr$ by setting $p=R\pr$.

\paragraph{Exogenous Poisson random measure}
In the remainder of the paper we assume that the  state process
solves the following SDE
\begin{eqnarray}
\label{filter_N}
 d\pr_t&=&\alphar(\pr_t,h_t) dt+ \betar\tr(\pr_t) dB_t +
 \int_{\mathcal{U}} \gammar(\pr_{t-},u) \widetilde{N}(dt, du),
\end{eqnarray}
where $\alphar$ and $\betar$ are defined above, $\gammar:
\simplexr\times \mathcal{U}\to \R^{d-1}$, and $\widetilde{N}$ is the
compensated measure to  some finite activity Poisson random measure
$N$ with jumps in a set $\mathcal{U} \subset \R^\kappa$. The
compensator of $N$ is denoted by  $\nu(du) \lambda dt$, i.e. we have
$\widetilde N(dt , du) = N(dt, du)-\nu(dz)\lambda dt$. In the next
section we show that for a proper choice of $ \gammar (\cdot)$ and
$\widetilde{N}(dt, du)$  the solution of \eqref{filter_N} has the  same
law as the original state process from \eqref{filter_restr_I}.

In order to ensure that  SDE \eqref{filter_N} has for each control
$h\in\mathcal{H}$ a unique strong solution and  for the proof of
some of the  estimates in Section \ref{regularization} the
coefficients $\alphar, \betar$ and $\gammar$  have to satisfy
certain Lipschitz and growth conditions (see
\cite{bib:Jacod-Shiriaev-03} and \cite{Pham (1998)}). These
conditions are given below. For technical reasons we require that
the conditions hold not only for~$\pr\in\simplexr$ but also for a
slightly larger set $\simplexr_\eps \supset \simplexr$ defined for
sufficiently small $\eps\ge 0$ by
\[\simplexr_\eps := \{\pr \in \R^{d-1}: \dist(\pi, \simplexr)\le \eps\},\]
where we denoted the distance of $\pi \in \R^{d-1}$ to $\simplexr$
by $ \dist(\pi, \simplexr) := \inf\{|\pr-\pr_0|_\infty \colon
\pr_0\in \simplexr\} $, for $|\pr|_\infty$ the maximum norm on
$\R^{d-1}.$

\begin{assumption}[Lipschitz and growth conditions]
\label{coef-model}
There exist  constants \linebreak[4] $C_L,\epsmax>0$  and a function $\rho:
\mathcal{U}\to\R_+$ with
$\int_{\mathcal{U}}\rho^2(u)\nu(du)< \infty$  such that for all
$\pr_1,\pr_2\in\simplex_\eps$, $\eps< \epsmax$ and $k=1,\ldots,d$
\begin{eqnarray}
\label{lipsch:b,sigma}
\sup_{h\in\compactset}
\norm{\alphar(\pr_1,h)-\alphar(\pr_2,h)}+\norm{\betar_k(\pr_1)-\betar_k(\pr_2)}&\leq&
C_L\norm{\pr_1-\pr_2} ,
\\
\label{growth_cond} \norm{\alphar(\pr,h)}+\norm{\betar_k(\pr)} &\le&
C_L(1+\norm{\pr}),\\
\norm{\gammar(\pr_1,u)-\gammar(\pr_2,u)}&\leq& \rho(u)\norm{\pr_1-\pr_2}, \label{lipsch:Delta}\\
\norm{\gammar(\pr,u)}&\leq&\rho(u)(1+\norm{\pr}).\label{growth-delta}
\end{eqnarray}
\end{assumption}

In our case the coefficients $\alphar$ and $\betar$ are continuously
differentiable  functions of $\pr$  on the compact set
$\simplexr_\eps$ and $h\in \compactset$ is bounded. Hence, the
Lipschitz and growth condition \eqref{lipsch:b,sigma} and
\eqref{growth_cond} are fulfilled. Specific conditions on the
densities $f_k(\cdot)$ that guarantee \eqref{lipsch:Delta} and
\eqref{growth-delta} are given in the next section.

For  the optimization problem \eqref{eq:objective-1} we can give an equivalent formulation
in terms of the restricted state variable $\pi$
with dynamics given in \eqref{filter_N}, that is the equation  driven by an exogenous Poisson random measure. For this it is convenient to denote for a given strategy $h\in \mathcal{H}$ the solution of the SDE \eqref{filter_N} starting at time $t \le T$ in the state $\pi \in \simplexr$ by $\pr^{(t,\pi,h)}$.
 This control problem reads as
\begin{equation} \label{control_problem_restr}
\max \Big \{  E \Big(\exp\Big\{\int_0^T
-b(R\pr_s^{(0,\prini,h)},h_s;\theta) ds \Big\}\Big) \colon h \in
\mathcal{H} \Big \}.
\end{equation}
The associated reward and value function for $(t,\pr)\in [0,T]\times \simplexr$ are
\begin{eqnarray}
 \nonumber v(t,\pr,h) &=& E \Big(
\exp\Big\{\int_t^T -b ( R \pr_s^{(t,\pr,h)},h_s)\,
ds\Big\}\Big)
\quad \text{for } \,h\in \mathcal{H},\\[-1ex]
\label{value_restr}\\[-1ex]
\nonumber V(t,\pr) & = &\sup \{ v(t,\pr,h)\colon h \in
\mathcal{H}\}.
\end{eqnarray}
 The generator associated to the solution of the
 state equation \eqref{filter_N} reads as
\begin{eqnarray}\label{generator_restr}
\mathcal{L}^{h} g(\pr)&=&
\frac{1}{2}\sum_{i,j=1}^{d-1}\restr{\beta}_i^{\top}(\pr)\restr{\beta}_j(\pr)g_{\pr^i\pr^j}
+ \sum_{i=1}^{d-1} \restr{\alpha}^i(\pr,h) g_{\pr^i}
\\\nonumber && \hspace*{15mm}
+\int_{\mathcal{U}}\{g(\pr+\restr{\gamma}(\pr,u))-g(\pr)\}\nu(du)
\end{eqnarray}
and the associated dynamic programming equation is
\begin{eqnarray}
\label{PDE_restr}
V_t(t,\pr)+\sup_{h \in \compactset} \Big\{\mathcal{L}^{h}
V(t,\pr)-b(R\pr,h;\theta)V(t,\pr)\Big\}=0, ~(t,\pr) \in [0,T) \times
\simplexr.
\end{eqnarray}

\section{State Equation with Exogenous Poisson Random Measure}
\label{spec_example}
In this section we show how a solution of the state equation
\eqref{filter_restr_I} can be constructed by means of an SDE of the
form \eqref{filter_N} that is driven by an exogenous Poisson random
measure. The main tool for constructing $\gammar$ will be the
so-called inverse Rosenblatt or distributional transform, see
R\"uschendorf \cite{Rueschendorf (2009)}, which is an extension of
the quantile transformation to the multivariate case.

We impose the following regularity conditions on the functions
$f_k(\cdot)$ that represent the conditional densities of $Z_n$ given
$Y_{T_n} = e_k$.

\begin{assumption}\label{ass_density}
All densities $f_k(z)$, $1 \le k \le d$,  are continuously
differentiable and have the common support $\mathcal{Z}$. We assume
that  $\mathcal{Z}$ is a  $\kappa$-dimensional rectangle
$[a,b]\subset \R^\kappa$, i.e.
 $$ \mathcal{Z}=\{z\in\R^\kappa: -\infty < a_k\le z_k\le b_k <\infty,\; k=1,\ldots,\kappa\}.$$
Moreover, there is some  $0<C_1$ such that $f_k(z) > C_1$  for all
$z\in \mathcal{Z}, k=1,\ldots, d$.
 \end{assumption}

\begin{remark}
Examples for densities that satisfy the above assumption are easily
constructed. Start with $C^1$-densities  $\widetilde{f}_k$  and
choose some (large) rectangle $\mathcal{Z}$ then
\begin{equation}
\label{density1}f_k(z)=
(1-\eps)\frac{\widetilde{f}_k(z)}{\int_\mathcal{Z}
\widetilde{f}_k(u)du} 1_\mathcal{Z}(z) +\eps
\frac{1}{|\mathcal{Z}|}1_\mathcal{Z}(z)  \quad
\text{where}\quad|\mathcal{Z}| = \prod\limits_{i=1}^\kappa (b_i-a_i)
\end{equation}
$k=1,\ldots,\kappa, ~\eps\in (0,1]$, satisfy the requirements of
Assumption \ref{ass_density}. Intuitively, \eqref{density1}
corresponds to a mixture of the original density $\widetilde{f}_k$
and the uniform distribution. The latter distribution carries no
information so there is uniform  a lower bound on the information
carried by a single expert opinion.
\end{remark}

\paragraph{Inverse Rosenblatt Transform}
In order to write our state equation \eqref{filter_restr_I} in the
form \eqref{filter_N}  with exogenous Poisson random measure we
apply  the inverse Rosenblatt transform, see for instance see
\cite{Rueschendorf (2009)}. Denote by $\mathcal{U} = [0,1]^\kappa$
the unit cube in $\R^\kappa$. In our context the inverse Rosenblatt
transform is a mapping $G:\mathcal{U}\to\mathcal{Z} $ such that for
a uniform random variable $U $ on $[0,1]^\kappa$ the random variable
$Z= G(U)$ has density $\overline{f}(z,p) =  \sum_{j=1}^d p^j
f_j(z),~p=R\pr$; the mapping $G$   can thus be viewed as a
generalization of the well-known quantile transform.

Now we explain the construction of the transformation $G$ in detail.
First, we define  for
$k=1,\ldots,\kappa-1$, $p=R\pr$, the marginal densities
\begin{equation}
\label{f_margin}
 f_{Z_1\ldots Z_{k}}(z_1,\ldots,z_{k},p) =
\int_{a_{k+1}}^{b_{k+1}}\hspace*{-1.0em}\ldots
\!\int_{a_\kappa}^{b_\kappa}
\overline{f}(z_1,\ldots,z_{k},s_{k+1},\ldots, s_\kappa,p) ds_\kappa
\ldots ds_{k+1}.
\end{equation}
For $k=\kappa$ we set $f_{Z_1\ldots Z_\kappa} := \overline{f}$. Next
we define for $k=2,\ldots,\kappa$ the conditional densities
\begin{eqnarray*}
f_{Z_k|Z_1\ldots Z_{k-1}}(z_k|z_1,\ldots z_{k-1},p) &:=&
\frac{f_{Z_1\ldots Z_k}(z_1,\ldots,z_k,p) }{f_{Z_1\ldots
Z_{k-1}}(z_1,\ldots,z_{k-1},p) }
\end{eqnarray*}
and the associated  distribution functions
\begin{eqnarray*}
F_{Z_k|Z_1\ldots Z_{k-1}}(z_k|z_1,\ldots z_{k-1},p) &=&
\int_{a_k}^{z_k} f_{Z_k|Z_1\ldots Z_{k-1}}(s_k|z_1,\ldots z_{k-1},p)
ds_k;
\end{eqnarray*}
for $k=1$ we denote by $F_{Z_1}$ the distribution function of $Z_1$.
Now we introduce the Rosenblatt transform
$\widetilde{F}:\mathcal{Z}\to [0,1]^\kappa=\mathcal{U}$, $z \mapsto (\widetilde{F}_1(z,p),\ldots,\widetilde{F}_\kappa(z,p) )^\top$ by
\begin{equation}
\label{Fw_def}  \widetilde{F}_1(z,p) = F_{Z_1}(z_1,p)
 \text{ and }   \widetilde{F}_k(z,p) =
F_{Z_k|Z_1\ldots Z_{k-1}}(z_k|z_1,\ldots z_{k-1},p),\; k=2,\ldots,
\kappa.
\end{equation}
Clearly,  $\widetilde{F}_k(z,p)$ depends on the first $k$ variables
$z_1,\ldots, z_k$, only. The desired transformation $G$ will be the inverse of $\widetilde F$, and the explicit form  of $\widetilde{F}$ is needed when we estimate the derivatives of $G$ in the proof of Lemma~\ref{gamma_lipschitz} below.

Assumption
\ref{ass_density} ensures that the joint density $\overline{f}(z,p)$
is finite and bounded away from zero. Hence, the conditional
densities $f_{Z_k|Z_1\ldots Z_{k-1}}(z|z_1,\ldots z_{k-1},p)$  are
strictly positive,  and the mapping $z \mapsto F_{Z_k|Z_1\ldots Z_{k-1}}(z|z_1,\ldots z_{k-1},p)$ is  strictly increasing and hence invertible. In the sequel we denote the  corresponding inverse function  by
$F^{-1}_{Z_k|Z_1\ldots Z_{k-1}}(\cdot|z_1,\ldots z_{k-1},p)$.

Now the desired transformation $Z= G(U)=G(U,p)$ with
transformation function $G:\mathcal{U} \to \mathcal{Z}$, $u \mapsto (G_1(u,p),\ldots,
G_\kappa(u,p))^\top$   can be
defined recursively by
\begin{equation} \label{G_def}
\begin{split}
G_1(u,p) &=F_{Z_1}^{-1}(u_1,p)\,,  \text{ and for }k=2,\ldots,\kappa,
\\
 G_k(u,p) &= F_{Z_k|Z_1\ldots Z_{k-1}}^{-1}\big(u_k\mid G_1(u,p),\ldots, G_{k-1}(u,p)\big)
\end{split}
\end{equation}
Note, that by construction it holds  $G(\widetilde{F}(z,p),p) = z$.
From \cite{Rueschendorf (2009)} it is known that for $U$ is
uniformly distributed in $[0,1]^\kappa$,  the random vector
$Z=(Z_1,\ldots,Z_\kappa)^\top=G(U,p)$ has the joint distribution
density $\overline f(z,p)$.

With the transformation $G$ at hand we  define the jump coefficient
$\gammar(\pr,u)$ by
\begin{equation}
\label{gammaN_def}
 \gammar^k(\pi,u) =
\pi^k\Big(\frac{f_k(G(u,R\pr))}{\overline f(G(u,R\pr),R\pr)} -1\Big)
~~\text{for } u\in\mathcal{U},~~k=1,\ldots,d-1.
\end{equation}
Moreover, we choose the  Poisson random measure $ N(dt, du) $ in
\eqref{filter_N} such that the associated  compound Poisson process
has constant intensity $\lambda$ and jump heights which are
uniformly distributed on $\mathcal{U}=[0,1]^\kappa$. Then, the
compensator of $N$ is $\nu(du)dt=\lambda du\, dt$ and the
compensated measure reads as $\widetilde{N}(dt, du) = N(dt, du) -
\lambda du\, dt$. Note that with this definition  the solution
$\pr_t$ of (\ref{filter_N}) satisfies for some Borel set $ A \subset
\R^{d-1}$
\begin{eqnarray*}
P(\Delta \pr_{T_n} \in A\mid \mathcal{F}_{T_n -}) &=&  \int_{\mathcal{U}} 1_A \big(\gammar(\pr_{T_n -},u)\big) du \\
&=& \int_{\mathcal{U}} 1_A \Big(\gammaIr\big(\pr_{T_n -}, G(u,R\pr_{T_n -})\big)\Big) du\\
&=& \int_{\mathcal{Z}} 1_A \big(\gammaIr (\pr_{T_n -},z)\big)\,\overline{f}(z,R\pr_{T_n -}) \, dz\,.
\end{eqnarray*}
 Hence  with the above choice of $\gammaN$ and $N(dt, du)$,  for constant $h$ the process $R\pr$, $\pr$ the solution of the SDE~\eqref{filter_N}, solves the martingale problem associated to the generator $\mathcal{L}^h$ from \eqref{new-generator}. Below we show that under Assumption~\ref{ass_density} the Lipschitz and growth conditions from~Assumption~\ref{coef-model} hold, so that the SDE~\eqref{filter_N} has a unique solution. It is well-known that this implies that the martingale problem associated with  $\mathcal{L}^h$ has a unique solution, see for instance Jacod \& Shiriaev \cite{bib:Jacod-Shiriaev-03}, Theorem~III.2.26. Hence $R\pr$  has the same law as the state variable process $p$ in~\eqref{filter_h}, which shows that we have achieved the desired  reformulation of the dynamics of the problem in terms of an exogenous Poisson random measure.

\begin{remark} Admittedly, the construction of $G$ and $\gammar$ is quite involved. The main reason for this  is the fact that we consider the case of multidimensional expert opinions with values in $\R^\kappa$ for some $\kappa >1$. Note however, that such a multivariate situation arises naturally in a model with more than one risky asset.
\end{remark}

\paragraph{Lipschitz and growth conditions} The next Lemma  states that under Assumption~\ref{ass_density} the functions $\gammar^k(\pi,u)$ satisfy the Lipschitz and growth conditions \eqref{lipsch:Delta} and \eqref{growth-delta}. The proof is given in Appendix \ref{proof_gamma_lipschitz}.
\begin{lemma}
\label{gamma_lipschitz} Under Assumption \ref{ass_density} and for
$\eps \le \epsmax:=\frac{C_1}{(d-1)C_2}$ the coefficient
$\gammar(\pr,u)$ defined in \eqref{gammaN_def} satisfies for $\pr\in
\simplexr_\eps$ the Lipschitz and growth condition
\eqref{lipsch:Delta} and \eqref{growth-delta}.
\end{lemma}

\section{Viscosity Solution}
\label{viscosity_sol}

In this section we show that the value function of the control problem~\eqref{value_restr} is a viscosity solution of the dynamic programming equation~\eqref{PDE_restr}. Since it is known from the literature that the  comparison principle holds for these equation (a precise reference is given below) we obtain an interesting characterization of the value function as viscosity solution of \eqref{PDE_restr}. This part of our analysis is based to a large extent on the work of  Pham~\cite{Pham (1998)}.

\paragraph{Preliminaries} The following estimates are crucial in proving that the
value function $V(t,p)$ is a viscosity solution of  \eqref{PDE_restr}.
\begin{proposition}%
\label{estimates} For any $k\in[0,2]$ there exists a constant $C>0$
such that for all $\delta\geq 0,t\in [0,T]$, $\pr,\xi\in\simplexr$,
$h\in \mathcal{H}$ and all stopping times $\tau$ between $t$ and
$T\wedge t+\delta$
\begin{eqnarray}
E_{}\big(~ \big| \pr_{\tau}^{(t,\pr,h)}\big|^k ~\big) &\leq& C(1+\norm{\pr}^k) \label{esti1}\\
E_{}\big(~\big|\pr_{\tau}^{(t,\pr,h)}-\pr\big|^k~\big)&\leq& C(1+\norm{\pr}^k)\delta^{\frac{k}{2}}\label{esti2}\\
E_{}\Big(\Big\{\sup_{t\leq s\leq t+\delta}\norm{\pr_{s}^{(t,\pr,h)}-\pr} \Big\}^k\Big)&\leq& C(1+\norm{\pr}^k) \delta^{\frac{k}{2}}\label{esti3}\\
E_{}\big(~ \big| \pr_{\tau}^{(t,\pr,h)}-\pr_{\tau}^{(t,\xi,h)}
\big|^k~\big)&\leq& C\norm{\pr-\xi}^2\label{esti4}.
\end{eqnarray}
\end{proposition}
\noindent The proof is given in  Appendix \ref{proof_estimates}.

Next we  state the dynamic programming principle associated to the
control problem \eqref{control_problem_restr}.
\begin{proposition}[Dynamic Programming Principle]\label{bellman-principle}
For $t\in [0,T]$, $\pr\in\simplexr$ and every  stopping time
$\delta$ such that $0\leq \delta\leq T-t$ we have
\begin{eqnarray*}
V(t,\pr)&=&\sup_{h\in \mathcal{H}} E \Big(
\exp\Big\{\int_t^{t+\delta} -b(R\pr_s^{(t,\pr,h)},h_s)ds
\Big\}V\big(t+\delta,\pr_{t+\delta}^{(t,\pr,h)}\big)\Big)
\end{eqnarray*}
\end{proposition}
\noindent For the proof of dynamic programming  principle we refer to Pham
\cite{Pham (1998)}, Proposition~3.1.

Applying  the dynamic programming principle yields the next
proposition on the continuity of the value function. The proof is
given in Appendix \ref{proof_prop_continuity}.
\begin{proposition}
\label{prop_continuity} There exists a constant $C>0$ such that for
all $t,s\in[0,T]$ and $\pr_1,\pr_2\in\simplexr$
\begin{eqnarray}\label{continuity}
|V(t,\pr_1)-V(s,\pr_2)|\leq
C\big[(1+|\pr_1|)|t-s|^{\frac{1}{2}}+|\pr_1-\pr_2|\big].
\end{eqnarray}
\end{proposition}

\paragraph{Viscosity Solution}
Following Pham \cite{Pham (1998)} we adapt the notion
of a viscosity solution introduced by Crandall and Lions
\cite{Crandall and Lions (1983)} to the case of integro-differential equations.
This concept consists  in  interpreting
equation \eqref{PDE_restr}  in a weaker sense. To simplify notation we split
the generator $\mathcal{L}^{h}$ given in \eqref{generator_restr} into
\[\mathcal{L}^{h} g(\pr)= \mathcal{A}^{h}g(\pr)+\mathcal{B}(\pr)\] where for $g\in C^2 (\simplexr)$
 the linear second-order differential operator $\mathcal{A}^h$ is defined by
\[
\mathcal{A}^{h}
g(\pr)=\frac{1}{2}\sum_{i,j=1}^{d-1}\betar_i^{\top}(\pr)\betar_j(\pr)g_{\pr^i\pr^j}(\pr)
+ \sum_{i=1}^{d-1}\alphar^i(\pr,h) g_{\pr^i} (\pr)
\]
and $\mathcal{B}$ is the integral operator
\[
\mathcal{B}g(\pr)=\lambda\int_{\mathcal{U}}\{g(\pr+\gammar(\pr,u))-g(\pr)\}\nu(du).
\]
Moreover $D_\pr g$ and $D^2_\pr g$ denote the gradient and Hessian
matrix of $g$ w.r.t~$\pr$.
\begin{definition}
\label{vis_defi}
(1)  A function  $V\in\mathcal{C}^0([0,T]\times\mathcal{S})$ is a
viscosity {\em supersolution} ({\em subsolution}) of equation
\eqref{bell-equa0} if
\begin{eqnarray}\label{sup-sub-sol}
-\frac{\partial \psi}{\partial t}(\tq,\piq)
-\sup_{h\in\compactset}\Big(-b(R\piq,h)V(\tq,\piq)+\mathcal{A}^{h}\psi(\tq,\piq)\Big)-\mathcal{B}\psi(\tq,\piq)\ge
0
\end{eqnarray}
(resp. $\leq 0$) for all $(\tq,\piq)\in [0,T]\times \simplex$ and
for all $\psi\in\mathcal{C}^{1,2}([0,T]\times\mathcal{S})$ with Lipschitz continuous derivatives $\psi_t, D_\pr^2\psi$ such that
$(\tq,\piq)$ is a global minimizer (maximizer) of the difference
$V-\psi$ on $[0,T]\times \simplex$ with $V(\tq,\piq)=\psi(\tq,\piq)$.
\vspace{0.1cm}

\noindent (2) $V$ is a viscosity solution of \eqref{bell-equa0} if it
is both super and subsolution of that equation.
\end{definition}

\begin{proposition}[Viscosity solution]
\label{viscosity} The value function $V(t,\pr)$ associated to the
optimization problem (\ref{eq:objective-1}) is a viscosity solution
of $(\ref{bell-equa0})$
\end{proposition}
\begin{proof}

\paragraph{Supersolution inequality}
Let be $\psi$ such that
\begin{eqnarray}\label{supsolution}
0=(V-\psi)(\tq,\piq)=\min_{[0,T]\times \mathcal{S}}(V-\psi).
\end{eqnarray}
We apply the dynamic programming principle for a fixed time
$\delta\in [0,T-t]$ to get
\begin{equation*}
V(\tq,\piq)=\psi(\tq,\piq)=\sup_{h\in \mathcal{H}} E \Big(
\exp\Big\{\int_\tq^{\tq+\delta} -b(R\pr_s^{(\tq,\piq,h)},h_s)ds
\Big\}V(\tq+\delta,\pr^{(\tq,\piq,h)}_{\tq+\delta})\Big).
\end{equation*}
From $(\ref{supsolution})$ we obtain
\begin{equation}
\label{vis1} 0 \geq \sup_{h\in \mathcal{H}} E \Big(
\exp\Big\{\int_\tq^{\tq+\delta} -b(R\pr_s^{(\tq,\piq,h)},h_s)ds
\Big\}\psi(\tq+\delta,\pr^{(\tq,\piq
,h)}_{\tq+\delta})-\psi(\tq,\piq)\Big).
\end{equation}
We now define for $u\in [\tq,T]$
\begin{equation}
\label{eta_Z} \eta_u:= \exp\Big\{\int_\tq^{u}
-b(R\pr_s^{(\tq,\piq,h)},h_s)ds\Big\} \quad\text{and}\quad
Z_u:=\eta_u\, \psi(u,\pr_{u}^{(\tq,\piq ,h)}).
\end{equation}
 Then, we apply It\^{o}'s formula to
$Z_{\tq+\delta}$, where we use the shorthand  notation  $\pr_{s}$
for $\pr_{s}^{(\tq,\piq,h)}$. Since $dZ_t = - b(R\pr_t,h_t)\eta_t\psi(t,\pr_{t})dt + \eta_t d\psi(t,\pr_{t})$, we have
\begin{eqnarray*}
Z_{\tq+\delta} &=&\psi(\tq,\piq
)+\int_\tq^{\tq+\delta}-b(R\pr_s,h_s)\eta_s\psi(s,\pr_{s})ds\\
&&+\int_\tq^{\tq+\delta}\eta_s\big\{\psi_t(s,\pr_{s})+\mathcal{A}^h\psi(s,\pr_{s})
+ \mathcal{B}\psi(s,\pr_{s})\big\}ds\\
&&+\int_\tq^{\tq+\delta}\eta_sD_\pr\psi(s,\pr_{s})\beta^{\top}(\pr_{s})dB_s
\\&&+\int_\tq^{\tq+\delta}\eta_s\int_\mathcal{U}
\big(\psi(s,\pr_{s}+
\gamma(\pr_{s},u))-\psi(s,\pr_{s})\big)\widetilde{N}(ds\times du).
\end{eqnarray*}
Due to our assumptions on $b$ and $\psi$, the last two  terms are martingales with zero expectations. From \eqref{vis1} we therefore  obtain
\begin{eqnarray}
\nonumber
0&\geq&\sup_{h\in \mathcal{H}} E(Z_{\tq+\delta} - \psi(\tq,\piq))\\
\label{eq1} &=&\sup_{h\in \mathcal{H}} E\Big(\int_\tq^{\tq+\delta}
-b(R\pr^{(\tq,\piq
,h)}_s,h_s)\eta_s\psi(s,\pr^{(\tq,\piq,h)}_{s})ds\\
\nonumber && \hspace*{8mm} + \int_\tq^{\tq+\delta}
\eta_s\{\psi_t(s,\pr^{(\tq,\piq,h)}_{s}) +
\mathcal{A}^{h}\psi(s,\pr^{(\tq,\piq,h)}_{s})+
\mathcal{B}\psi(s,\pr^{(\tq,\piq,h)}_{s})\}ds \Big).
\end{eqnarray}
We now show for the first integral that
\begin{eqnarray}
\nonumber \lefteqn{
E\Big(\int_\tq^{\tq+\delta}  -b(R\pr^{(\tq,\piq,h)}_{s},h_s)\eta_s\psi(s,\pr^{(\tq,\piq,h)}_{s})ds\Big) \geq }  \hspace*{30mm}\\
\label{eq2} &  & E\Big(\int_\tq^{\tq+\delta}
-b(R\pr,h_s)\psi(\tq,\piq)ds\Big)-\delta\varepsilon(\delta).
\end{eqnarray}
where $\varepsilon(\delta)\to 0$ as $\delta\to 0$. Using the
Lipschitz continuity of $\psi$ we obtain the inequality
$
|\psi(s,\pr^{(\tq,\piq,h)}_{s})-\psi(\tq,\piq)|\leq  C(
|s-\tq|+|\pr^{(\tq,\piq ,h)}_{s}-\piq |)
$,
which leads to
\begin{eqnarray*}
\lefteqn{\sup_{h\in \mathcal{H}}
E\Big(\int_\tq^{\tq+\delta}-b(R\pr^{(\tq,\piq,h)}_s,h_s)\eta_s\psi(s,\pr^{(\tq,\piq,h)}_{s})\Big)ds\geq}
\\
&& \hspace*{-3mm}
\sup_{h\in \mathcal{H}}
E\Big(\int_\tq^{\tq+\delta}-b(R\pr^{(\tq,\piq
,h)}_s,h_s)\eta_s\psi(\tq,\piq)ds\Big) -C\delta \Big\{\delta+
E\Big(\sup_{0\leq s\leq
\delta}|\pr^{(\tq,\piq,h)}_{s}-\piq|\Big)\Big\}\Big).
\end{eqnarray*}
By Proposition \ref{estimates} we have  $E\Big(\sup_{0\leq s\leq
\delta}|\pr^{(\tq,\piq ,h)}_{s}-\piq|\Big)\leq C (1+\norm{\piq
})\delta^{\frac{1}{2}}$ and hence we obtain
\begin{eqnarray}
\nonumber
\lefteqn{E\Big(\int_\tq^{\tq+\delta}-b(R\pr^{(\tq,\piq,h)}_s,h_s)\eta_s\psi(s,\pr^{(\tq,\piq,h)}_{s})ds\Big)\geq} \hspace*{25mm}\\
\label{eq2a}
&&E\Big(\int_\tq^{\tq+\delta}-b(R\pr^{(\tq,\piq,h)}_s,h_s)\eta_s\psi(\tq,\piq
)ds\Big)-\delta\varepsilon(\delta).
\end{eqnarray}
Recall from (\ref{eq:def-b}) that  $b(R\pr,h) = -\theta\Big(h^{\top} M R \pr - \frac{1-\theta}{2}
\norm{\sigma^{\top} h }^2\Big)$. Since this expression  depends
linearly on $\pr$ and $h_s$ and since $h_s$ takes values in the compact set
$\compactset$ we have
\[
|b(R\pr^{(\tq,\piq,h)}_s,h_s)-b(R\pr,h_s)|\leq
C|\pr^{(\tq,\piq,h)}_s-\pr|.
\]
Using $|\eta_s-\eta_\tq|=|\eta_s-1|\le C|s-\tq|$ and
 the same computations to get (\ref{eq2a}) it yields
\[E\Big(\int_\tq^{\tq+\delta} \hspace*{-4mm}
-b(R\pr^{(\tq,\piq,h)}_s,h_s)\eta_s\psi(s,\pr^{(\tq,\piq
,h)}_{t+s})ds\Big) \geq E\Big(\int_\tq^{\tq+\delta} \hspace*{-4mm}
-b(R\piq,h_s)\psi(\tq,\piq)ds\Big)-\delta\varepsilon(\delta).\]
Applying similar  computations to the other terms in $(\ref{eq1})$
by using the estimates for the state process $\pr$ and the Lipschitz
continuity of $D^2_\pr \psi$ we obtain
\begin{eqnarray*}
\varepsilon(\delta)&\geq&\frac{1}{\delta}\sup_{h\in \mathcal{H}}
E\Big(\int_\tq^{\tq+\delta}\{-b(\piq,h_s)\psi(\tq,\piq)+
\psi_t(\tq,\piq)+ \mathcal{A}^{h}\psi(\tq,\piq)+
\mathcal{B}\psi(\tq,\piq)\}ds\Big).
\end{eqnarray*}
Replacing $h\in\mathcal{H}$ by a constant strategy in the above sup
we get
\[
\varepsilon(\delta) \geq
\frac{1}{\delta}\Big(\int_\tq^{\tq+\delta}\sup_{h\in\compactset}\Big(-b(R \piq,h_s)\psi(\tq,\piq)+\psi_t(\tq,\piq)
+\mathcal{A}^{h}\psi(\tq,\piq)+ \mathcal{B}\psi(\tq,\piq)\Big)ds.
\]
Applying the mean value theorem and  sending $\delta$ to $0$  we get the supersolution viscosity
inequality:
\begin{eqnarray*}
-\frac{\partial \psi}{\partial
t}(\tq,\piq)-\sup_{h\in\compactset}\Big(-b(R \piq,h)V(\tq,\piq)+A^{h}\psi(\tq,\piq)\Big)-\mathcal{B}\psi(\tq,\piq)\geq
0.
\end{eqnarray*}
\paragraph{Subsolution inequality}
Let $\psi$ be such that
\begin{eqnarray}\label{subsolution}
0=(V-\psi)(\tq,\piq)=\max_{[0,T]\times \mathcal{S}}(V-\psi)
\end{eqnarray}
As a consequence of dynamic programming  principle in Proposition
$\ref{bellman-principle}$ we have
\begin{eqnarray*}
V(\tq,\piq)&=&\sup_{h\in \mathcal{H}} E \Big(
\exp\Big\{\int_\tq^{\tq+\delta} -b(R\pr_s^{(\tq,\piq,h)},h_s)ds
\Big\}V(t+\delta,\pr^{(\tq,\piq ,h)}_{t+\delta})\Big)
\end{eqnarray*}
Equation $(\ref{subsolution})$ implies that
\begin{eqnarray*}
0&\leq&\sup_{h\in \mathcal{H}} E \Big(
\exp\Big\{\int_\tq^{\tq+\delta} -b(R \pr_s^{(\tq,\piq,h)},h_s)ds
\Big\}\psi(t+\delta,\pr^{(\tq,\piq,h)}_{t+\delta})-\psi(\tq,\piq)\Big).
\end{eqnarray*}
Using similar computations by applying It\^{o}'s formula to the
process $Z_u$ given in (\ref{eta_Z}) and using the estimates for the
state process $\pr$ we obtain
\begin{eqnarray*}
\varepsilon(\delta)&\leq&\frac{1}{\delta}\sup_{h\in \mathcal{H}}
E\Big(\int_\tq^{\tq+\delta} \{-b(R \piq,h_s)\psi(\tq,\piq)
+\psi_t(\tq,\piq) + \mathcal{A}^{h}\psi(\tq,\piq)
+\mathcal{B}\psi(\tq,\piq)\}ds\Big).\nonumber
\end{eqnarray*}
Replacing $h\in\mathcal{H}$ by a constant strategy in the above sup,
applying mean the value theorem and  sending $\delta$ to $0$  we obtain the subsolution viscosity inequality
\[
-\psi_t(\tq,\piq)-\sup_{h\in \compactset}\Big(-b(R\piq,h)V(\tq,\piq)+\mathcal{A}^{h}\psi(\tq,\piq)\Big)\nonumber\\
-\mathcal{B}\psi(\tq,\piq) \le 0.
\]
\end{proof}

\paragraph{Comparison principle} Here we quote the following result, which is Theorem~4.1 of \cite{Pham  (1998)}.
\begin{theorem} \label{thm:comparison} Suppose that Assumption~\ref{coef-model} holds and that $u_1$ and $u_2$ are continuous functions on $[0,T] \times \simplexr $ such that $u_1$ is a subsolution and $u_2$ is a supersolution of the dynamic programming equation (\ref{bell-equa0}). If $u_1(T,\pi) \le u_2(T,\pi)$ for all $\pi \in \simplexr$, then
$$ u_1(t,\pi) \le u_2(t,\pi) \text{ for all }  (t,\pi ) \in [0,T] \times \simplexr \,.$$
\end{theorem}

Together with Proposition~\ref{viscosity}, this result implies immediately that the value function $V(t,\pr)$ associated to the
optimization problem (\ref{eq:objective-1}) is the unique continuous viscosity solution
of $(\ref{bell-equa0})$.

\section{Regularized Dynamic Programming Equation}
\label{regularization} In this section we introduce the regularized
version of our  dynamic programming problem and we discuss the
convergence of reward and value function as the regularization-terms
converge to zero. In Corollary \ref{eps_opt} we finally show that
optimal strategies in the regularized problem are nearly optimal in
the original problem.

\paragraph{Regularized state equation}
Since regularization will drive the state process outside the set
$\simplexr$ we need to extend  the definition of the coefficients
$\restr{\alpha}, \restr{\beta}$ and $\restr{\gamma}$ from
$\simplexr$ to the whole $\R^{d-1}$. For $\pi\in \R^{d-1}$, $h\in
\compactset$ and $\eps>0$ we define
\[\widetilde{\restr{\alpha}}(\pr,h) := \left\{
\begin{array}{cl}
\restr{\alpha}(\pr,h) (1-\dist(\pi, \simplexr) /\eps)  &  ~~\text{for}\quad \pr \in \simplexr_\eps\\
0 & ~~\text{otherwise.}
\end{array}
\right. \] Note, that $\simplexr \subset \simplexr_\eps$ and there
is a continuous transition to zero if $\dist(\pi, \simplexr)$
reaches $\eps$. Moreover, on $\simplexr$ it holds
$\widetilde{\restr{\alpha}}(\pr,h) =\restr{\alpha}(\pr,h)$, i.e.~the
coefficients coincide. Analogously we define
$\widetilde{\restr{\beta}}$ and $\widetilde{\restr{\gamma}}$ as
extensions of ${\restr{\beta}}$ and ${\restr{\gamma}}$.
\begin{lemma}
\label{lipschitz_growth_restr} Under the assumptions of Lemma
\ref{coef-model} the coefficients $\widetilde{\restr{\alpha}},
\widetilde{\restr{\beta}}$ and $\widetilde{\restr{\gamma}}$  satisfy
the Lipschitz and growth conditions \eqref{lipsch:b,sigma} to
\eqref{growth-delta} for $\pr \in \R^{d-1}$.
\end{lemma}
\begin{proof}
The Lipschitz and growth conditions for the coefficients $\alphar,
\betar$ and $\gammar$  given in Lemma \ref{coef-model} hold for
$\pr\in \simplexr_\eps$ for $\eps\le\epsmax$. Multiplication of
these functions  by the bounded and Lipschitz continuous function
$1-\dist(\pi, \simplexr) /\eps$ preserves the Lipschitz and growth
property.
\end{proof}
For the sake of simplicity of notation in the sequel we will
suppress the tilde and simply write $\restr{\alpha}, \restr{\beta}$
and $\restr{\gamma}$ instead of $\widetilde{\restr{\alpha}},
\widetilde{\restr{\beta}}$ and $\widetilde{\restr{\gamma}}$.

Next we define the dynamics of the regularized state process ${}^\np
\pr_t$
\begin{equation}
\label{filter_pert} d\,{}^\np \pr_t = \restr{\alpha}({}^\np
\pr_t,h_t) dt + \restr{\beta}\tr({}^\np \pr_t) dB_t +
\int\nolimits_{\mathcal{U}} \restr{\gammaN}({}^\np \pr_{t-},u)
\widetilde{N}(dt, du) +\frac{1}{\sqrt{\np}} d\widetilde B_t
\end{equation}
where $\widetilde B_t$ denotes a $d-1$-dimensional Brownian motion
independent of $B_t$.  This state process is now  driven by an
$n+d-1$-dimensional Brownion motion.  Note that the diffusion
coefficient of the regularized equation $( \restr{\beta}\tr(\pr_t),
\frac{1}{\sqrt{\np}}I_{d-1} )\tr$ satisfies the Lipschitz and growth
condition (\ref{lipsch:b,sigma})  and \eqref{growth_cond}
 given in Lemma \ref{coef-model} since
$\restr{\beta}(\pr_t)$ satisfies these conditions  and
$\frac{1}{\sqrt{\np}}I_{d-1}$ does not depend on $p$.

\paragraph{$L_2$-Convergence ${}^\np\pi_t \to \pi_t$}\ We now compare the
solution ${}^\np \pr_t$ of the regularized state equation
(\ref{filter_pert}) with the solution $\pr_t$ of the unregularized
state equation \eqref{filter_N} and study asymptotic properties for
$\np\to \infty$. This will be crucial for establishing  convergence
of the associated reward function of the regularized problem to the
original optimization problem.

We assume that both processes start at time $t_0\in[0,T]$ with the
same initial value $q\in\simplexr$, i.e. ${}^\np
\pr_{t_0}=\pr_{t_0}=q $. The corresponding solutions are denoted by
${}^\np\pr_t^{(t_0,q,h)}$ and $\pr_t^{(t_0,q,h)}$.

\begin{lemma}[Uniform $L_2$-convergence
w.r.t.~$h\in\mathcal{H}$] \label{conv_filter} It holds for
$\np\to\infty$
\[E\Big( \sup_{t_0\le t\le T} \Big|{}^\np\pr_t^{(t_0,q,h)}-\pr_t^{(t_0,q,h)}\Big|^2\Big) \longrightarrow 0
\quad\text{uniformly for   } h\in \mathcal{H}.\]
\end{lemma}

\begin{proof}
To simplify the notation we suppress the superscript $(t_0,q,h)$ and
write $\pi_t$ and ${}^\np\pi_t$. Moreover, we denote by $C$ a
generic constant.

We give the proof for $t_0=0$, only. Using the corresponding
representation as stochastic integrals for the solutions of the
above SDEs we find
\begin{eqnarray}
\nonumber
{}^\np\pi_t - \pi_t &=& A_t^\np + M_t^\np \qquad  \text{where}\\
\nonumber\quad A_t^\np& :=& \int\nolimits_0^t
(\restr{\alpha}({}^\np\pr_s,h_s) - \restr{\alpha}(\pi_s,h_s))
ds\quad \text{and}
\\
\nonumber M_t^\np & = & \int\nolimits_0^t
(\restr{\beta}({}^\np\pr_s)-\restr{\beta}(\pr_s))\tr  dB_s +
\int\nolimits_0^t \int\nolimits_{\mathcal{U}}
(\restr{\gamma}({}^\np\pr_s,u)-\restr{\gamma}(\pr_s,u))\widetilde{N}(ds, du)\\
\nonumber &&+ \frac{1}{\sqrt{\np}} \,d\widetilde B_t.
\end{eqnarray}
Note that here we have used the fact that the SDE for ${}^\np\pr$ and for $\pr$ is driven by an exogenous Poisson random measure, since this permits us to write the difference of the jump-terms as stochastic integral with respect to the same compensated random measure.

Denoting $G_t^\np:=E\Big( \sup_{s\le t}|{}^\np\pr_s-\pr_s|^2\Big)$
it holds
\begin{equation}
\label{Gt_def} G_t^\np = E\Big( \sup_{s\le t}|A_s^\np + M_s^\np
|^2\Big) \le 2 E\Big( \sup_{s\le t}\norm{A_s^\np }^2\Big) + 2 E\Big(
\sup_{s\le t}\norm{M_s^\np }^2\Big).
\end{equation}
For the first term on the r.h.s.~we find by applying Cauchy-Schwarz
inequality and the Lipschitz condition (\ref{lipsch:b,sigma}) for
$\alpha$
\begin{eqnarray*}
\sup_{s\le t}\norm{A_s^\np }^2& = & \sup_{s\le t} \norm{
\int\nolimits_0^s
\big(\restr{\alpha}({}^\np\pr_u,h_u) - \restr{\alpha}(\pi_u,h_u) \big) du}^2\\
&\le& \sup_{s\le t} \; s\cdot \int\nolimits_0^s
\Big| \restr{\alpha}({}^\np\pr_u,h_u) - \restr{\alpha}(\pi_u,h_u)) \Big|^2 du\\
& \le & t \cdot \int\nolimits_0^t C_L |{}^\np\pr_u - \pr_u|^2 du \le
t \cdot \int\nolimits_0^t C_L \sup_{v\le u} |{}^\np\pr_v - \pr_v|^2
du.
\end{eqnarray*}
Note that the constant $C_L$ does not depend on $h$.
Taking expectation it follows
\begin{eqnarray}
\label{A_esti} E\Big( \sup_{s\le t}\norm{A_s^\np }^2\Big) &\le &
t\cdot C_L \int_0^t E\Big(\sup_{v\le s} |{}^\np\pr_v - \pr_v|^2
ds\Big)\le  C\int_0^t G_s^\np ds.
\end{eqnarray}
For the second term on the r.h.s.~of (\ref{Gt_def}) Doob's
inequality for martingales  yields
\begin{eqnarray}
\nonumber
 E\Big( \sup_{s\le t}\norm{M_s^\np }^2\Big) &\le & 4
E(\norm{M_t^\np}^2) \\
\label{M_esti}&=& 4  \Big( \int\nolimits_0^t
E\big(\operatorname{tr}[(\restr{\beta}({}^\np\pr_s)-\restr{\beta}(\pr_s))\tr
(\restr{\beta}({}^\np\pr_s)-\restr{\beta}(\pr_s))] \big)ds \\
\nonumber
 && +   \int\nolimits_0^t \int\nolimits_{\mathcal{U}}
E\big(\norm{\restr{\gamma}({}^\np\pr_s,u)-\restr{\gamma}(\pr_s,u))}^2\big)
\nu(du) ds + \frac{(d-1)t}{\np}\Big).
\end{eqnarray}
Using the Lipschitz conditions  (\ref{lipsch:b,sigma}) and
(\ref{lipsch:Delta}) for the coefficients $\beta$ and $\gamma$ it
follows
\begin{eqnarray*}
\lefteqn{E\big(tr[(\restr{\beta}({}^\np\pr_s)-\restr{\beta}(\pr_s))\tr
(\restr{\beta}({}^\np\pr_s)-\restr{\beta}(\pr_s))]\big)
\le  C_L^2 E(\norm{{}^\np\pr_s - \pr_s}^2)}\hspace*{50mm}\\
 &\le &C_L^2 E\big(\sup_{v\le s} \norm{{}^\np\pr_v - \pr_v}^2\big)  = C_L^2 G_s^\np\\
 E\big(\norm{\restr{\gamma}({}^\np\pr_s,u)-\restr{\gamma}(\pr_s,u))}^2\big) &\le &
 \rho^2(u)E(\norm{{}^\np\pr_s - \pr_s}^2)\\
 &\le &\rho^2(u) E\big(\sup_{v\le s} \norm{{}^\np\pr_v - \pr_v}^2\big) = \rho^2(u) G_s^\np.
\end{eqnarray*}
Substituting the above estimates into (\ref{M_esti}) it follows that
\begin{eqnarray}
\nonumber E\Big( \sup_{s\le t}\norm{M_s^\np }^2\Big) &\le & 4  \Big(
\int\nolimits_0^t C_L^2 G_s^\np ds +
 \int\nolimits_0^t G_s^\np  ds\int\nolimits_{\mathcal{U}} \rho^2(u)
 \nu(du)
+ \frac{(d-1)t}{\np}\Big)\\
\label{M_esti2} &\le & C \int\nolimits_0^t  G_s^\np ds +
\frac{4(d-1)t}{\np}.
\end{eqnarray}
Substituting (\ref{A_esti}) and (\ref{M_esti2}) into (\ref{Gt_def})
we find
\[G_t^\np \le \frac{4(d-1)T}{m} + C \int\nolimits_0^t G_s^\np ds .\]
Finally we apply Gronwall Lemma to derive
\[G_T^\np \le \frac{4(d-1)T}{m}\; e^{CT} \to 0 \quad\text{for } \np\to \infty\]
which concludes the proof.

\end{proof}

Note, that the $L_2$-convergence  for the restricted state process
${}^\np \pr_t$ established in Lemma \ref{conv_filter} also holds for
the associated $d$-dimensional process ${}^\np p_t=R{}^\np \pr_t$.

We now extend the notions of reward and value function given in
\eqref{reward_value_function} to the  process ${}^\np p_t=R\,{}^\np
\pr_t$ with ${}^\np \pr_t$ satisfying the regularized state equation
\eqref{filter_pert}. Since ${}^\np p_t$ takes values in $\R^{d}$ (and
not only in $\simplex$) we extend the function $b$ given in
\eqref{eq:def-b} to $p=R\pr\in \R^d$. With the notation $b_*=
\min\{b(p,h),~p\in\simplex, h\in \compactset\}$ \text{and}
  $b^*= \max\{b(p,h),~p\in\simplex, h\in \compactset\}$
we define  $$\widetilde{b}(p,h) :=(b(p,h) \vee b_*) \wedge b^*.$$
Then $\widetilde{b}$ is  bounded on $\R^d\times \compactset$ and for
$p\in \simplex$ the  function $\widetilde{b}$ coincides with $b$.
In the sequel we simply write $b$
instead of $\widetilde{b}$.
We define the reward and value function associated to the
regularized state equation \eqref{filter_pert} by
\begin{eqnarray}
\label{reward_m} \nonumber v^m(t,\pr,h) &=& E \Big(
\exp\Big\{\int_t^T -b ( R({}^\np\pr_s^{(t,\pr,h)}),h_s)
ds\Big\}\Big)
\quad \text{for } h\in \mathcal{H},\\
\nonumber V^m(t,\pr) & = &\sup \{ v^m(t,\pr,h)\colon h \in
\mathcal{H}\}.
\end{eqnarray}
Recall that $v(t,\pr,h) $ and
$V(t,\pr)$ defined in \eqref{value_restr} denote the  reward and value function associated to the unregularized state
equation \eqref{filter_N}. The generator associated to the solution of the
regularized state equation \eqref{filter_pert} reads as
\begin{eqnarray*}\label{generator_pert}
{}^\np\mathcal{L}^{h} g(\pr)&=&
\frac{1}{2}\sum_{i,j=1}^{d-1}\restr{\beta}_i^{\top}(\pr)\restr{\beta}_j(\pr)g_{\pr^i\pr^j}
+ \frac{1}{2\np} \sum_{i=1}^{d-1} g_{\pr^i\pr^i}
+ \sum_{i=1}^{d-1} \restr{\alpha}^i(\pr,h) g_{\pr^i}\\
&& \hspace*{15mm}
+\int_{\mathcal{U}}\{g(\pr+\restr{\gamma}(\pr,u))-g(\pr)\}\nu(du)
\end{eqnarray*}
and the associated dynamic programming equation is
\begin{equation}\label{eq:HJB-regularized}
V^m_t(t,\pr)+\sup_{h \in \compactset} \Big\{{}^\np\mathcal{L}^{h}
V^m(t,\pr)-b(R\pr,h;\theta)V^m(t,\pr)\Big\}=0, ~(t,\pr) \in [0,T)
\times \R^{d-1}.
\end{equation}
Note, that for the  generator ${}^\np\mathcal{L}^{h}$  the
\textit{ellipticity condition}  for the coefficients of the second
derivatives holds: we have for all $z\in \R^{d-1}\setminus\{0\}$
\begin{eqnarray*}
z\tr(\restr{\beta}\tr\restr{\beta} + \frac{1}{2m} I_{d-1}) z &=&
z\tr \restr{\beta}\tr\restr{\beta} z + \frac{1}{2m}z\tr z =
~|\restr{\beta} z|^2 ~+\frac{1}{2m} |z|^2 > 0.
\end{eqnarray*}
Hence the results of
Davis \&  Lleo \cite{Davis and Lleo (2012)} apply to this dynamic programming problem. equation. According to Theorem~3.8 of their paper,  there is a classical solution $V^\np $ of \eqref{eq:HJB-regularized}. Moreover, for every $(t, \pi)$ there is a unique maximizer ${}^\np h^*$  of the problem
$$\sup_{h \in \compactset} \Big\{{}^\np\mathcal{L}^{h}
V^m(t,\pr)-b(R\pr,h;\theta)V^m(t,\pr)\Big\},$$
${}^\np h^*$ can be chosen as a Borel-measurable function of $t$ and $\pi$ and the optimal strategy is given by ${}^\np h^*_t = {}^\np h^*(t, {}^\np \pr_t) $; see also the discussion
preceding Remark~\ref{rem:strategy}.

\paragraph{Convergence of reward and value function}
The next theorem on the uniform convergence of reward functions is
our main result; convergence of the value function and
$\varepsilon$-optimality of ${}^\np h^*$ follow easily from this
theorem.

\begin{theorem}[Uniform Convergence of reward functions]
\label{conv_reward}  It holds
\[
\sup\limits_{h\in\mathcal{H}} |v^\np(t,\pr,h) - v(t,\pr,h)|\to 0
\quad\text{for } \np\to\infty, \quad t\in[0,T],~\pr\in \simplexr.\]
\end{theorem}

\begin{proof}
We introduce the notation
\[J:= \int_t^T -b (R\pr_s^{(t,\pr,h)},h_s) ds \quad\text{and}\quad
J^\np:= \int_t^T -b (R\,{}^\np \pr_s^{(t,\pr,h)},h_s) ds.\] Then the
reward functions read as $v(t,\pr,h)=E ( e^{J})$ and
$v^\np(t,\pr,h)=E ( e^{J^\np})$ and it holds
\begin{eqnarray}
\nonumber  |v^\np(t,\pr,h) - v(t,\pr,h)| &=& |E ( e^{J^\np} -e^{J})
| \le E (
|e^{J^\np} -e^{J}|) \\
\label{cc1} &\le &   C E(|J^\np-J| ),
\end{eqnarray}
where we used Lipschitz continuity of $f(x)=e^x$ on bounded
intervals and the boundedness of $J$ and $J^\np$ which follows from
the the boundedness of $b$.
Using Lipschitz continuity of $b$ we derive
\begin{eqnarray}
\nonumber
 E(|J^\np -J|) &= & E\Big(\Big|\int_t^T\big[ b (R\pr_s^{(t,\pr,h)},h_s)-b (R\,{}^\np \pr_s^{(t,\pr,h)},h_s)\big] ds \Big|\Big)\\
\nonumber
 &\le &  \int_t^T C\, E(|{}^\np \pr_s^{(t,\pr,h)}- \pr_s^{(t,\pr,h)}| )\, ds\\
&\le &  C \int_t^T \Big(E(|{}^\np \pr_s^{(t,\pr,h)}-
\pr_s^{(t,\pr,h)}|^2\Big)^{1/2} \, ds \to 0 \label{cc3}
\end{eqnarray}
for $\np\to\infty$ and uniformly w.r.t.~$h\in\mathcal{H}$ which
follows from Lemma \ref{conv_filter}. Plugging (\ref{cc3}) into
(\ref{cc1}) we find
\[\sup\limits_{h\in\mathcal{H}} |v^\np(t,\pr,h) - v(t,\pr,h)|\to 0
\quad\text{for } \np\to\infty.\]

\end{proof}
\begin{corollary}[Convergence of value functions]
\label{conv_value}
 It holds
\[V^\np(t,\pr) \to V(t,\pr)
\quad\text{for } \np\to\infty, \quad t\in[0,T], ~\pr\in \simplexr.\]
\end{corollary}

\begin{proof}
For $\theta\in(0,1)$ the assertion follows from
\begin{eqnarray*}
|V^\np(t,\pr) - V(t,\pr)| &= & \Big| \sup\limits_{h\in\mathcal{H}}
v^\np(t,\pr,h) - \sup\limits_{h\in\mathcal{H}} v(t,\pr,h) \Big|\\
& \le&
\sup\limits_{h\in\mathcal{H}} |v^\np(t,\pr,h) - v(t,\pr,h)|
\end{eqnarray*}
and Lemma \ref{conv_reward}.
Analogously, for $\theta<0$ it follows
\begin{eqnarray*}
|V^\np(t,\pr) - V(t,\pr)| &= & \Big| \inf\limits_{h\in\mathcal{H}}
v^\np(t,\pr,h) - \inf\limits_{h\in\mathcal{H}} v(t,\pr,h) \Big|
\\&= &
\Big| \sup\limits_{h\in\mathcal{H}} (-v^\np(t,\pr,h)) -
\sup\limits_{h\in\mathcal{H}} (-v(t,\pr,h))
\Big|\\
&\le& \sup\limits_{h\in\mathcal{H}} |v^\np(t,\pr,h) - v(t,\pr,h)|.
\end{eqnarray*}
\end{proof}

\paragraph{On $\varepsilon$-optimal stratgies}
 Finally we show that the optimal strategy ${}^\np h^*$ for the regularized problem is $\varepsilon$-optimal in the original problem. This gives a method for computing (nearly) optimal strategies.
\begin{corollary}[$\varepsilon$-optimality]
\label{eps_opt} For every $\varepsilon>0$ there exists some
$\np_0\in\N$ such that
\[|V(t,\pr)- v(t,\pr,{}^\np h^*)| \le \varepsilon\quad\text{for }\np\ge \np_0,\]
i.e.~${}^\np h^*$ is an $\varepsilon$-optimal strategy for the
original control problem.
\end{corollary}
\begin{proof}
It holds
\begin{eqnarray}
\nonumber\lefteqn{\hspace*{-15mm} |V(t,\pr) -v(t,\pr,{}^\np h^*)|}\\
\nonumber
& \le& |V(t,\pr) -v^\np(t,\pr,{}^\np h^*)| + |v^\np(t,\pr,{}^\np h^*)- v(t,\pr,{}^\np h^*)|  \\
\label{eps1} &= & |V(t,\pr) -V^\np(t,\pr)| + |v^\np(t,\pr,{}^\np
h^*)- v(t,\pr,{}^\np h^*)|
\end{eqnarray}
where for the first term on the r.h.s.~we used $v^\np(t,\pr,{}^\np
h^*) = V^\np(t,\pr)$. Using the convergence properties for the
reward function given in Lemma \ref{conv_reward} and for the value
function given in Corollary \ref{conv_value}  we can find
 for every $\varepsilon>0$  some $\np_0\in \N$ such that for $\np\ge
 \np_0$ it holds
 \[|V(t,\pr) -V^\np(t,\pr)| \le \frac{\varepsilon}{2} \quad\text{and}\quad
 |v^\np(t,\pr,{}^\np h^*)- v(t,\pr,{}^\np h^*)| \le \frac{\varepsilon}{2}.\]
 Plugging  the above estimates into (\ref{eps1}) it follows for $\np\ge \np_0$
 \begin{eqnarray}
\nonumber |V(t,\pr) -v(t,\pr,{}^\np h^*)| &\le&
\frac{\varepsilon}{2} + \frac{\varepsilon}{2} = \varepsilon.
\end{eqnarray}
\end{proof}

\begin{remark}
Note that in the proof of the corollary we use that the sequence of
reward functions $v^\np$ converges to $v$ \textit{uniformly} in $h$.
This is a stronger property than convergence of the value functions
$V^\np$ to $V$ so that standard stability results for dynamic
programming equations are not sufficient to proof the corollary.
\end{remark}

\bibliographystyle{amsplain}

\appendix

\section{Proof of Lemma \ref{gamma_lipschitz}}
\label{proof_gamma_lipschitz}
\begin{proof}
We give the proof for the maximum norm $\normmax{.}$ in $\R^d$. From
this the assertion for the Euclidean norm can be deduced from the
equivalence of norms.

Note that the fact that all densities are $\mathcal{C}^1$ with
compact support $\mathcal{Z} = [a,b]$ implies the existence of
constants $C_2, C_d < \infty $ such that  for all $1 \le k \le d$,
$z \in \mathcal{Z}$,
\begin{equation}
\label{density_bounds} f_k(z) \le  C_2  \text{ and }
 \Big|\frac{\partial}{\partial z_i}
f_j(z)\Big| \le C_d,\quad i=1,\ldots,\kappa.
\end{equation}

\paragraph{Boundedness of $\overline{f}(z,R\pr)$}
First we show that for $\pr\in \simplexr_\eps, \; z\in \mathcal{Z}$
and $\eps< \epsmax$ there are constants $0 <C_*\le  C^*< \infty $
such that
\begin{equation}
\label{fquer_bounds}  C_* \le \overline{f}(z,R\pr) \le C^*.
\end{equation}
For this, observe that for $p = R \pi$,
\begin{equation}
\label{fquer_h1} \overline{f}(z,p) = \sum_{j=1}^d p^j f_j(z) =
\sum_{p^j< 0} p^j f_j(z) + \sum_{p^j\ge 0} p^j f_j(z).
\end{equation}
For the lower bound we deduce
\begin{eqnarray*}
\overline{f}(z,p) &\ge& \sum_{p^j< 0} (-\eps) \max_j f_j(z) + \sum_{p^j\ge 0} p^j \min_j f_j(z)\\
&\ge & -\eps(d-1) C_2 \quad + \quad \Big(1-\sum_{p^j< 0}
p^j\Big)\cdot C_1 \ge  -\eps\frac{C_1}{\epsmax} +  1\cdot C_1 = C_*,
\end{eqnarray*}
where we  used   Assumption \ref{ass_density}, \eqref{density_bounds}, $p^j\ge -\eps$ and
$\sum_{j=1}^d p^j=1$. For the upper bound from \ref{fquer_h1} we
find
\begin{eqnarray*}
\overline{f}(z,p) &\le& 0 + \sum_{p^j\ge 0} p^j \max_j f_j(z)\le
 \Big(1-\sum_{p^j< 0} p^j\Big) C_2 \le  (1+\eps(d-1)) C_2
 =C^*
\end{eqnarray*}
Note that the lower bound in \eqref{fquer_bounds}  implies that
$\overline f(\cdot,R\pr) $ is strictly positive for $\pr\in
\simplexr_\eps$. Moreover, since the components of $p = R\pr$ sum up
to one by definition $\overline f(\cdot,R\pr) $ is a strictly
positive probability density for $\pr\in \simplexr_\eps$. Hence, the
inverse Rosenblatt transform $G(u,R\pr)$ and thus the function
$\gammar(\pr,u)$ defined in \eqref{gammaN_def} is well defined  for
$\pr\in \simplexr_\eps$ (and not just for $\pr\in \simplexr$).

\paragraph{Proof of the Lipschitz condition \eqref{lipsch:Delta}}
Clearly, \eqref{lipsch:Delta} holds for some   constant function
$\rho(u)=\overline \rho$ if we can show that the derivatives of
$\gammar(\pr, u)$ with respect to $\pr^j$ are bounded for all $1 \le
j \le d-1$.    This is obviously  equivalent to estimating the
derivatives of
$$ \gamma^k(p,u) = p^k\Big(\frac{f_k(G(u,p))}{\overline f(G(u,p),p)} -1\Big)$$
with respect to the components~$p^j$ where $p=R\pi$. Let
$$
c^k_j(p,u) := \frac{\partial}{\partial p^j}
\Big(\frac{f_k(G(u,p))}{\overline f(G(u,p),p)} -1\Big),\quad
j,k=1,\ldots,d.
$$
Then it holds
$$\frac{\partial}{\partial p^j} \gammaN^k(p,u) = \delta_{jk}  \Big(\frac{f_k(G(u,p))}{\overline f(G(u,p),p)}
-1\Big) + p^k c^k_j(p,u).$$
The first term on the r.h.s. is bounded since it holds for
$k=1,\ldots,d$ and $\eps<\epsmax$
\begin{equation}
\label{h_esti} \frac{f_k(G(u,p))}{\overline f (G(u,p),p)} \le
\frac{C_2}{C_*},
\end{equation}
where we have used  \eqref{density_bounds} and the lower bound for
$\overline f(z,p)$ given in (\ref{fquer_bounds}).

It remains to show that $c^k_j(p,u)$ is bounded. Abbreviating
$z=z(p)=G(u,p)$ we find
\begin{eqnarray}
\nonumber c^k_j(p,u) &=& \frac{1}{(\overline{f}(z,p))^2}\bigg(
\sum\limits_{l=1}^\kappa
\frac{\partial}{\partial z_l} f_k(z) \frac{\partial}{\partial p^j} G_l(u,p) \cdot \overline{f}(z,p)  \\
&& \label{Lip_h2} \hspace*{15mm} - f_k(z) \cdot
\Big(f_j(z)+\sum_{i=1}^d p^i \sum\limits_{l=1}^\kappa
\frac{\partial}{\partial z_l} f_i(z) \frac{\partial}{\partial p^j}
G_l(u,p) \Big) \bigg).~~~
\end{eqnarray}
Using  \eqref{density_bounds}, $\sum_{j=1}^d |p^j|\le
 1+(d-1)\epsmax  $ and estimate (\ref{fquer_bounds}) for $\overline{f}$,  we derive
for $\eps<\epsmax$
\begin{eqnarray}
\nonumber |c^k_j(p,u)| &\le & \frac{1}{C_*^2} \Big(
C_d \sum\limits_{l=1}^\kappa \Big|\frac{\partial}{\partial p^j} G_l(u,p)\Big| C^*  \\
\label{c_esti1} && + C_2 \cdot \Big(C_2+ (1+(d-1)\epsmax)  C_d
\sum\limits_{l=1}^\kappa \Big|\frac{\partial}{\partial p^j}
G_l(u,p)\Big| \;\Big) \Big).
\end{eqnarray}
In Lemma~\ref{G_deriv_bounded} below we show that the derivatives $\frac{\partial}{\partial
p^j} G_l(u,p)$ are bounded, that is there is some $C>0$ such that for
$j=1,\ldots,d$ and $l=1,\ldots,\kappa$
$
\big|\frac{\partial}{\partial p^j} G_l(u,p) \big| \le C.
$
From this the boundedness of $c^k_j$ follows immediately.

\paragraph{Proof of the growth  condition \eqref{growth-delta}}
Here we apply estimate (\ref{h_esti}) and find
\begin{eqnarray*}
|\gamma^j(p,u)|  = \Big|p^j \Big(\frac{f_j(G(u,p))}{\overline f
(G(u,p),p)} -1\Big)\Big| &\le & |p^j| \Big(\frac{C_2}{C_*} +1\Big)
 \le  (1+\normmax{p})\Big(\frac{C_2}{C_* }
+1\Big)
\end{eqnarray*}
and hence $\normmax{\gamma(p,u)}\leq \overline \rho(1+\normmax{p})$
with some constant $\overline \rho$.
\end{proof}

\begin{lemma}
\label{G_deriv_bounded} Under the assumptions of Lemma
\ref{gamma_lipschitz} there exists a constant $C>0$ such that for
$j=1,\ldots,d$ and $l=1,\ldots,\kappa$
\[ \Big|\frac{\partial}{\partial
p^j} G_l(u,p) \Big| \le C.
\]
\end{lemma}

\begin{proof}
We derive from differentiating $G_l(\widetilde F(z,p),p)=z_l$
w.r.t.~$p_j$ using the chain rule
\[
 \sum\limits_{i=1}^\kappa \frac{\partial}{\partial u_i}
G_l(\widetilde F(z,p),p) \frac{\partial}{\partial p^j} \widetilde
F_i(z,p) + \frac{\partial}{\partial p^j} G_l(\widetilde F(z,p),p)=
0.\] Substituting $u= \widetilde F(z,p)$ we obtain the estimate
\begin{equation}
\label{G_deriv} \Big| \frac{\partial}{\partial p^j} G_l(u,p)\Big|
\le \sum\limits_{i=1}^\kappa \Big| \frac{\partial}{\partial u_i}
G_l(u,p)\Big| \Big|\frac{\partial}{\partial p^j} \widetilde
F_i(z,p)\Big|.
\end{equation}
(i) For the proof of the boundedness of the derivatives on the
r.h.s.~we need the following auxiliary estimates for the marginal
densities $f_{Z_1\ldots Z_{k}}, k=1,\ldots,\kappa$ given in
\eqref{f_margin}. From estimate \eqref{fquer_bounds} for
$\overline{f}$ we derive the estimate
\begin{eqnarray}
\label{f_margin_esti1}  C_* \prod\limits_{i=k+1}^\kappa (b_i-a_i)
&\le&  f_{Z_1\ldots Z_{k}}(z_1,\ldots,z_{k},p)~~ \le ~~C^*
\prod\limits_{i=k+1}^\kappa (b_i-a_i).
\end{eqnarray}
For the derivatives of the marginal densities w.r.t.~$p^j$ the
definition of $\overline{f}$ in \eqref{beta_gamma_def} yields
\begin{eqnarray*}
\frac{\partial}{\partial p^j} f_{Z_1\ldots
Z_{k}}(z_1,\ldots,z_{k},p) &=&
\int_{a_{k+1}}^{b_{k+1}}\hspace*{-1em}\ldots
\!\int_{a_\kappa}^{b_\kappa} f_j(z_1,\ldots,z_{k},s_{k+1},\ldots,
s_\kappa) ds_{k+1}\ldots ds_\kappa.
\end{eqnarray*}
From Assumption \ref{ass_density} and \eqref{density_bounds} it
follows
\begin{eqnarray}
\label{f_margin_esti2} 0< C_1 \prod\limits_{i=k+1}^\kappa (b_i-a_i)
\le \frac{\partial}{\partial p^j} f_{Z_1\ldots
Z_{k}}(z_1,\ldots,z_{k},p) \le C_2 \prod\limits_{i=k+1}^\kappa
(b_i-a_i).
\end{eqnarray}
For the derivatives of the marginal densities
w.r.t.~$z_j,~j=1\ldots,k$ we find
\begin{eqnarray}
\nonumber \lefteqn{
\Big| \frac{\partial}{\partial z_j} f_{Z_1\ldots Z_{k}}(z_1,\ldots,z_{k},p) \Big|  }\\
\nonumber &\le& \int_{a_{k+1}}^{b_{k+1}}\hspace*{-1em}\ldots
\!\int_{a_\kappa}^{b_\kappa} \sum\limits_{l=1}^d p^l
\Big|\frac{\partial}{\partial z_j}
f_l(z_1,\ldots,z_{k},s_{k+1},\ldots, s_\kappa)\Big|
ds_{k+1}\ldots ds_\kappa \\
\label{f_margin_esti3} &\le&   C_d \prod\limits_{i=k+1}^\kappa
(b_i-a_i),
\end{eqnarray}
where  the upper bound from \eqref{density_bounds} on the
derivatives of the densities $f_j$ has been used.
\\[1ex]
(ii) Now we can prove the boundedness  for the second term r.h.s.~of
\eqref{G_deriv}. For $k=2,\ldots,\kappa$ we obtain from the
definition of $\widetilde F(z,p)$ in \eqref{Fw_def}
\begin{eqnarray*}
 \Big|\frac{\partial}{\partial p^j} \widetilde
F_i(z,p)\Big| &=& \Big| \int_{a_k}^{z_k} \frac{\partial}{\partial
p^j} f_{Z_k|Z_1\ldots Z_{k-1}}(s_k|z_1,\ldots z_{k-1},p)
ds_k \Big|\\
&=& \Big| \int_{a_k}^{z_k} \frac{\partial}{\partial p^j}
\frac{f_{Z_1\ldots Z_k}(z_1,\ldots,z_{k-1},s_k,p) }{
      f_{Z_1\ldots Z_{k-1}}(z_1,\ldots,z_{k-1},p) }
ds_k \Big|\\
&\le & \int_{a_k}^{z_k} \frac{1}{f^2_{Z_1\ldots Z_{k-1}}(\cdot)}
\Big(~ \Big|\frac{\partial}{\partial p^j} f_{Z_1\ldots
Z_k}(\cdot)\Big| f_{Z_1\ldots Z_{k-1}}(\cdot)+
\\
&& \hspace*{30mm}   f_{Z_1\ldots Z_k}(.)
\Big|\frac{\partial}{\partial p^j} f_{Z_1\ldots Z_{k-1}}(.)
\Big|~\Big) ds_k \le  C.
\end{eqnarray*}
Here, we have used estimate \eqref{f_margin_esti1}, which states
that
 the marginal densities are bounded from above and bounded away
from zero, and \eqref{f_margin_esti2} for the boundedness of the
derivatives of the marginal densities w.r.t.~$p^j$.

For $k=1$ we observe that
\[
\frac{\partial}{\partial p^j} \widetilde F_1(z,p) =
\frac{\partial}{\partial p^j}  F_{Z_1}(z_1,p) = \int_{a_1}^{z_1}
\frac{\partial}{\partial p^j}  f_{Z_1}(s_1,p) ds_1 .
 \]
The boundedness  ${\partial}/{\partial p^j} \widetilde F_1(z,p)$ is
a consequence of  estimate \eqref{f_margin_esti2}.
\\[1ex]
(iii) For proving the boundedness of ${\partial}/{\partial u_i}
G_l(u,p)$ in \eqref{G_deriv} we consider the Jacobian matrices for
$G(u)$ and $\widetilde F(z)$ defined by
\[
J^G(u) := \Big( \frac{\partial}{\partial u_j} G_i(u,p)
\Big)_{i,j=1\ldots,\kappa} \quad \text{and} \quad J^{\widetilde
F}(z) := \Big( \frac{\partial}{\partial z_j} \widetilde F_i(z,p)
\Big)_{i,j=1\ldots,\kappa}.
\]
Below we show that
 for $z=G(u,p)$ the matrix  $J^{\widetilde F}(z)$ is
regular, hence   $J^G(u)=J^{-1}_{\widetilde F}(G(u,p))$, since
$G(\widetilde{F}(z,p),p)=z$. From the definition of $\widetilde F$
in \eqref{Fw_def} it follows that $J^{\widetilde F}(z)$ is a lower
triangular matrix since $\widetilde F_k$ depends on $z_1,\ldots,z_k$
only.

Next we consider the diagonal elements of $J^{\widetilde F}(z)$.
Using  \eqref{f_margin_esti1} we find constants $\underline C$ and
$\overline C$ such that $\underline C \le f_{Z_1\ldots
Z_k}(z_1,\ldots,z_k,p) \le \overline C$ for all $k=1\ldots,\kappa$.
Then it holds with $\delta:= \min\{\underline C,\,{\underline
C}\,/\,{\overline C}\}$
\begin{eqnarray*}
\frac{\partial}{\partial z_1} \widetilde F_1(z,p) &=&
 f_{Z_1}(z_1,p) \ge \delta \quad \text{and}\\
\frac{\partial}{\partial z_k} \widetilde F_k(z,p) &=&
f_{Z_k|Z_1\ldots Z_{k-1}}(z_k|z_1,\ldots z_{k-1},p) =
\frac{f_{Z_1\ldots Z_k}(z_1,\ldots,z_k,p) }{f_{Z_1\ldots
Z_{k-1}}(z_1,\ldots,z_{k-1},p) } \ge  \delta,
\end{eqnarray*}
for $k=2,\ldots,\kappa$. Since $J^{\widetilde F}(z)$ is triangular,
its determinant is
\[\det(J^{\widetilde F}(z) ) = \prod\limits_{k=1}^\kappa
\frac{\partial}{\partial z_k} \widetilde F_k(z,p) \ge \delta^\kappa
>0,\]
hence $J^{\widetilde F}(z)$ is invertible.

Next we show that the the non-zero off-diagonal elements of
$J^{\widetilde F}$ are bounded. It holds for $k=2,\ldots,\kappa,\;
j=1,\ldots, k-1$
\begin{eqnarray*}
 \frac{\partial}{\partial z_j} \widetilde F_k(z,p) &=&
\int_{a_k}^{z_k}  \frac{\partial}{\partial z_j} f_{Z_k|Z_1\ldots Z_{k-1}}(s_k|z_1,\ldots z_{k-1},p) d s_k\\
&= & \int_{a_k}^{z_k}  \frac{\partial}{\partial z_j}
\frac{f_{Z_1\ldots Z_k}(z_1,\ldots,z_{k-1}, s_k,p) }{f_{Z_1\ldots
Z_{k-1}}(z_1,\ldots,z_{k-1},p)
} d s_k\\
&\le & \int_{a_k}^{z_k} \frac{1}{f^2_{Z_1\ldots Z_{k-1}}(\cdot)}
\Big( ~\Big|\frac{\partial}{\partial z_j} f_{Z_1\ldots
Z_k}(\cdot)\Big| f_{Z_1\ldots Z_{k-1}}(\cdot)
\\
&& \hspace*{25mm}+    f_{Z_1\ldots Z_k}(\cdot)
\Big|\frac{\partial}{\partial z_j} f_{Z_1\ldots Z_{k-1}}(\cdot)
\Big|~\Big) ds_k\le  C.
\end{eqnarray*}
Here again we have used that
 the marginal densities are bounded from above and bounded away
from zero, and \eqref{f_margin_esti3} for the boundedness of the
derivatives of the marginal densities w.r.t.~$z_j$.

For proving the boundedness  of ${\partial}/{\partial u_i} G_l(u,p)$
in \eqref{G_deriv} which are the entries of the Jacobian matrix
$J^G(u)$ we use that $J^G$ is the inverse of $J^{\widetilde F}$.
Since $J^{\widetilde F}$ ist a triangular matrix the entries of
$J^G$ can be computed recursively by Gaussian elimination starting
with the first row. This gives that  for $k,l=1,\ldots,\kappa$
\[J^G_{kl} = \frac{1}{J^{\widetilde F}_{kk}}\Big( \delta_{kl} -
\sum\limits_{j=1}^{k-1} J^{\widetilde F}_{kj}\; J^G_{jl}
 \Big),\]
i.e.~the entry $J^G_{kl}$ can be represented by an affine linear
combination of the bounded off-diagonal entries in row $k$ of
$J^{\widetilde F}$ divided by $J^{\widetilde F}_{kk}$. The latter is
strictly positive and bounded from below by $\delta>0$. Hence, all
entries of $J^G$ are bounded.

\end{proof}
\section{Proof of Proposition \ref{estimates}}
\label{proof_estimates}
\begin{proof}
We give the proof for $k=2$. The assertions for $k\in[0,2]$ follow
from H\"older inequality. We denote by $C$ a generic constant.

\paragraph{Proof of inequality (\ref{esti1}): $E_{}(|\pr_{\tau}^{(t,\pr,h)}|^2)\leq
C(1+\norm{\pr}^2)$} \ \\
 We recall the state equation
\begin{equation}
\label{dim-filter_h}
\begin{split}
d\pr_t &=  \alphar(\pr_t,h_t)dt+\betar^{\top}(\pr_t)dB_t
+\int_{\mathcal{U}}\gammar(\pr_t,u)\widetilde{N}(dt\times du)
\end{split}
\end{equation}
and for the sake of shorter notation we denote  by
$\pr_\tau=\pr_{\tau}^{(t,\pr,h)}$ the solution of equation
$(\ref{dim-filter_h})$ starting from $\pr$ at time $t$ using strategy
$h$ for $\tau\ge t$.
 Then it holds
\begin{eqnarray*}
|\pr_{\tau}^{}|^2 &\leq & C\Big(\norm{\pr}^2+\Big|\int_t^{\tau}
\alphar(\pr_s,h_s)ds\Big|^2
+\Big|\int_t^{\tau}\betar(\pr_s)dB_s\Big|^2\\
&&\hspace*{45mm}+\Big|\int_t^{\tau}\int_{\mathcal{U}}\gammar(\pr_s,u)\widetilde{N}(ds\times
du)\Big|^2 \Big)
\end{eqnarray*}
Taking expectation and using It\^{o}-Levy isometry
 implies
\begin{eqnarray*}
E_{}(|\pr_{\tau}^{}|^2) &\leq &
C\Big(\norm{\pr}^2+E_{}\Big(\int_t^{\tau}| \alphar(\pr_s,h_s)|^2 ds\Big) + E_{}\Big(\int_t^{\tau} tr(\betar\tr(\pr_s) \betar(\pr_s)) ds\Big)\\
&&\hspace*{45mm} +
E_{}\Big(\int_t^{\tau}\int_{\mathcal{U}}|\gammar(\pr_s,u)|^2\nu(du)
ds)\Big) \Big).
\end{eqnarray*}

 We now use the linear growth of $\alphar$, $\betar$  and
$\gammar$ and the integrability property for $\rho$ (see Assumption
\ref{coef-model})  to obtain
\begin{eqnarray}
\nonumber E_{}(|\pr_{\tau}^{}|^2)&\leq& C\Big\{\norm{\pr}^2
+E_{}\Big(\int_t^{\tau}(1+|\pr_{s}^{}|^2)ds\Big)\Big\}
\le
 C\Big\{\norm{\pr}^2+E_{}(\tau) +E_{}\Big(\int_t^{\tau}|\pr_{s}^{}|^2ds\Big)\Big\}\\
&\leq&C\Big\{\norm{\pr}^2 +1+
E_{}\Big(\int_t^{\tau}|\pr_{s}^{}|^2ds\Big)\Big\}\label{ine-0}.
\end{eqnarray}
For any deterministic time $\tau=u$ Fubini's Theorem gives
\[E_{}(|\pr_{u}^{}|^2) \leq C\Big\{\norm{\pr}^2 +1+ \int_t^{u}E_{}(|\pr_{s}^{}|^2)ds\Big\}\]
and applying Gronwall's Lemma  to $G_u:=E_{}(|\pr_{u}^{}|^2)$ implies
\[
E_{}(|\pr_{u}^{}|^2)\leq C(\norm{\pr}^2+1) e^{C(u-t)} \le
C(\norm{\pr}^2+1).
\]
Finally, we note, that for any stopping time $\tau\in [t,T\wedge
t+\delta]$ it holds
\[E_{}\Big(\int_t^{\tau}|\pr_{s}^{}|^2ds\Big) \le \int_t^{t+\delta} E_{}(|\pr_{s}^{}|^2) ds \le C(1+\norm{\pr}^2).\]
Substituting the upper estimate back into $(\ref{ine-0})$ proves the
assertion.

\paragraph{Proof of inequality (\ref{esti2}): $E_{}(|\pr_{\tau}^{(t,\pr,h)}-\pr|^2) \leq C(1+\norm{\pr}^2)\delta$} \ \\
 The process $(\pr_{\tau}^{}-\pr)$ starts from $0$ and hence the
computations for $\pr_{\tau}^{}$ in the above proof inequality
\eqref{esti2}  give for $\tau\in [t,T\wedge t+\delta]$
\begin{eqnarray*}
E_{}(|\pr_{\tau}^{}-\pr|^2) &\leq& C\int_t^{\tau}(1+E_{}
(|\pr_{s}^{}|^2))ds
\leq C\int_t^{t+\delta}(1+E_{}(|\pr_{s}^{}|^2))ds\\
&\leq&  C\int_t^{t+\delta}(1+C(1+(|\pr|^2))ds \leq
C(1+\norm{\pr}^2)\delta.
\end{eqnarray*}

\paragraph{Proof of inequality (\ref{esti3}):
$E_{}\Big(\Big\{\sup_{t\leq s\leq
t+\delta}|\pr_{s}^{(t,\pr,h)}-\pr|\Big\}^2\Big) \leq
C(1+\norm{\pr}^2)\delta$}\ \\
We give the proof for $t=0$ from which the claim for general $t$
follows immediately. Using the corresponding representation as
stochastic integrals for the solution of equation
(\ref{dim-filter_h}) we find
\begin{eqnarray*}
\pr_s-\pr &=& A_s + M_s \qquad \text{where}\\
 A_s &=& \int_0^{s} \alphar(\pr_r,h_r)dr
~~\text{and}~~ M_s = \int_0^{s} \betar\tr(\pr_r)dB_r + \int_0^{s}
\gammar(\pr_r,u) {\widetilde N}(dr\times du).
\end{eqnarray*}
Then it holds
\begin{eqnarray}
\nonumber E_{}\big(\big\{\sup_{0\leq s\leq
\delta}|\pr_s-\pr|\big\}^2\big) &=& E_{}\big(\big\{\sup_{0\leq s\leq \delta}|A_s + M_s|\big\}^2\big)\\
\label{estim1} &\le & 2 E_{}\big(\sup_{0\leq s\leq \delta}|A_s|^2
\big) +   2 E_{}\big(\sup_{0\leq s\leq \delta}|M_s|^2 \big).
\end{eqnarray}
For the first term on the r.h.s.~we find by applying Cauchy-Schwarz
inequality and the growth condition (\ref{growth_cond}) for $\alphar$
\begin{eqnarray*}
\sup_{0\leq s\leq \delta}|A_s|^2& = & \sup_{0\leq s\leq \delta}
\norm{ \int_0^s \alphar(\pr_r,h_r) dr}^2
\le \sup_{0\leq s\leq \delta} s  \int_0^s \norm{\alphar(\pr_r,h_r)}^2 dr.\\
&\le &  \delta  \int_0^\delta C(1+\norm{\pr_r}^2) dr.
\end{eqnarray*}
Taking expectation and applying estimate (\ref{esti1}) we find
\begin{equation}
\label{estim2} E\big(\sup_{0\leq s\leq \delta}|A_s|^2\big) \le
\delta \int_0^\delta C(1+\norm{\pr}^2) dr \le \delta
C(1+\norm{\pr}^2).
\end{equation}
For the second term on the r.h.s. of (\ref{estim1}) Doob's
inequality for martingales and It\^{o}-Levy isometry yields
\begin{eqnarray*}
\nonumber
 E\Big( \sup_{0\leq s\leq \delta}\norm{M_s }^2\Big) \le  4
E(\norm{M_\delta}^2) &=& 4  \Big( \int_0^\delta
E\big(tr[\betar\tr(\pr_r)\betar(\pr_r)] \big)dr \\
&&\hspace*{5mm}+  \int_0^\delta \int_{\mathcal{U}}
E\big(\norm{\gammar(\pr_r,u)}^2\big)\nu(du) dr \Big).
\end{eqnarray*}
Applying the growth conditions (\ref{growth-delta}),
(\ref{growth_cond})   and estimate (\ref{esti1}) it yields
\begin{eqnarray}
\nonumber
 E\Big( \sup_{0\leq s\leq \delta}\norm{M_s }^2\Big) &\le&
 C  \Big( \int_0^\delta E(1+\norm{\pr_r}^2)dr +
 \int_0^\delta \int_{\mathcal{U}}
\rho^2(u) E(1+\norm{\pr_r}^2)\nu(du) dr \Big)\\
\label{estim3} &\le & C(1+\norm{\pr}^2) \int_0^\delta \!\!\Big( 1+
\!\!\int_{\mathcal{U}} \rho^2(u)\nu(du) \Big)dr \le
C\delta(1+\norm{\pr}^2).
\end{eqnarray}
Substituting (\ref{estim2}) and (\ref{estim3}) into (\ref{estim1})
yields the assertion.

\paragraph{Proof of inequality (\ref{esti4}):
$E_{}(|\pr_{\tau}^{(t,\pr,h)}-\pr_{\tau}^{(t,\xi,h)}|^2) \leq (\pr-\xi)^2$}\ \\
For the sake of shorter notation we write $\pr_s=\pr_s^{(t,\pr,h)}$ and
$\xi_s=\pr_s^{(t,\xi,h)}$ and we set
$\Delta{\alphar}(\pr,\xi,h)=\alphar(\pr,h) -\alphar(\xi,h),
~\Delta{\betar}(\pr,\xi)=\betar(\pr)-\betar(\xi)~  \text{ and }
\Delta{\gammar}(\pr,\xi)=\gammar(\pr,u)-\gammar(\xi,u).$ Then,
\begin{eqnarray*}
Y_\tau&:=&\pr_{\tau}-\xi_{\tau}=\pr-\xi+\int_t^{\tau}\Delta{\alphar}(\pr_{s},\xi_{s}^{},h_s)ds+
\int_t^{\tau}\Delta{\betar}^{\top}(\pr_{s}^{},\xi_{s}^{})dB_s
\\&&\hspace*{40mm}+\int_t^{\tau}\int_{\mathcal{U}}\Delta{\gammar}(\pr_{s}^{},\xi_{s}^{},u)\widetilde{N}(ds\times
du).
\end{eqnarray*}
Applying It\^{o}'s lemma to $Y_s^2$ and using It\^{o}-Levy isometry
  we obtain
\begin{eqnarray*}
E_{}(|Y_{\tau}|^2)\!\!&=& \!\!|\pr-\xi|^2+E\Big(\! \int_t^{\tau}\Big\{2
Y_s^{\top}\Delta{\alphar}(\pr_{s}^{},\xi_{s}^{},h_s)
+tr\Big(\Delta{\betar}(\pr_{s}^{},\xi_{s}^{})\Delta{\betar}\tr(\pr_{s}^{},\xi_{s}^{})\Big)\\
&&\hspace*{53mm}+\int_{\mathcal{U}}|\Delta
\gammar(\pr_{s}^{},\xi_{s}^{},u)|^2 \nu(du)\Big\}ds \Big).
\end{eqnarray*}
Hence we obtain from the Lipschitz continuity of $\alphar, \betar,
\gammar$ given in  Assumption $\ref{coef-model}$
\[
E_{}(|Y_{\tau}|^2)\leq |\pr-\xi|^2+CE_{}\Big(\int_t^{\tau}|Y_s|^2ds\Big)
.
\]
For any deterministic time $\tau=u$ Fubini's Theorem gives
\[
E_{}(|Y_{u}|^2)\leq |\pr-\xi|^2+CE_{}\Big(\int_t^{u}|Y_s|^2ds\Big)
\]
and applying Gronwall's Lemma  to $G_u:=E_{}(|Y_{u}|^2)$ implies
\[
E_{}(|Y_{u}|^2)\leq \norm{\pr-\xi}^2 e^{C(u-t)} \le  C\norm{\pr-\xi}^2.
\]
Finally, we note, that for any stopping time $\tau\in [t,T\wedge
t+\delta]$ it holds
\[E_{}(|Y_{\tau}|^2)\leq |\pr-\xi|^2+CE_{}\Big(\int_t^{t+\delta}|Y_s|^2ds\Big) \le  C\norm{\pr-\xi}^2. \]
\end{proof}

\section{Proof of Proposition \ref{prop_continuity}}
\label{proof_prop_continuity}
\begin{proof}
\paragraph{Boundedness of $V$}
We recall that $V(t,\pr)=\sup_{h\in\mathcal{H}}v(t,\pr,h)$ where
\begin{eqnarray*}
 v(t,\pr,h)&=&E_{}\Big(\exp\Big\{\int_t^{T} -b(R\pr_s^{(t,\pr,h)},h_s) ds \Big\}\Big)\\
\text{with} \quad
  b(p,h) &=& -\theta\Big(h^{\top} Mp - \frac{1-\theta}{2} \norm{\sigma^{\top} h
  }^2\Big),
\end{eqnarray*}
and $\pr_s^{(t,\pr,h)}$ is the solution of the SDE (\ref{filter_N}) with
initial value $\pr_t=\pr$.

The function  $b$  is bounded, since it is  continuous and  $\pr\in
\simplexr$ and $h\in\compactset$ take values in compact sets, i.e.
$|b(R\pr,h)|\le C_b$ with some constant $C_b>0$.
 Hence $0\le v(t,\pr,h)\leq e^{C_b(T-t)}\leq e^{C_bT}$ for
all $h\in\mathcal{H}$ which implies that $0\le V(t,\pr)\leq e^{C_bT}$.

Note, that since the value function $V$ is bounded, it also
satisfies the \textsl{linear growth
condition} $V(t,\pr)\leq C(1+\norm{\pr})$ since $\norm{\pr}_\infty\leq 1$. \\
\paragraph{Lipschitz continuity in $\pr$} \
The reward function can be written  as\\
$v(t,\pr,h)=E ( e^{J(\pr)})$ where $J(\pr):= \int_t^{T}
-b(R\pr_s^{(t,\pr,h)},h_s)ds$. It holds for $\theta\in(0,1)$
\begin{eqnarray}
\nonumber |V(t,\pr)-V(t,\xi)| &=&\big|\sup_{h\in \mathcal{H}} E
(e^{J(\pr)}) - \sup_{h\in \mathcal{H}} E
 (e^{J(\xi)})\big| \le  \sup_{h\in \mathcal{H}} \big| E (e^{J(\pr)} - e^{J(\xi)}) \big|\\
\label{lip1} &\le & \sup_{h\in \mathcal{H}}  E (|e^{J(\pr)} -
e^{J(\xi)}|) \le  \sup_{h\in \mathcal{H}}  C E (|J(\pr) - J(\xi)|) ,
\end{eqnarray}
where we used Lipschitz continuitiy of $f(x)=e^x$ on bounded
intervals and the boundedness of $J(\pr)$ which follows, since
 $b$ is bounded.
For $\theta<0$ we use $V(t,\pr)=\inf_{h\in \mathcal{H}}E (e^{J(\pr)}) =
\sup_{h\in \mathcal{H}} -E (e^{J(\pr)})$ and apply  analogous
estimates.

Using that $b$ is linear in $\pr$ and that $h_t\in \compactset$ is
uniformly bounded we derive
\begin{eqnarray}
\nonumber
 E(|J(\pr) -J(\xi)|) &= & E\Big(\Big|\int_t^{T}\big[ b(R\pr_s^{(t,\xi,h)},h_s)-b(R\pr_s^{(t,\pr,h)},h_s)\big] ds \Big|\Big)\\
\nonumber
 &\le &  \int_t^{T} C\, E( | \pr_s^{(t,\pr,h)}- \pr_s^{(t,\xi,h)}| )\, ds\\
 \label{c3}
&\le & \!\! C\! \int_t^{T}\!\! \norm{\pr-\xi}^2 \, ds  \le C(T-t) \norm{\pr-\xi}^2 \le
C \norm{\pr-\xi},
\end{eqnarray}
for every $h\in\mathcal{H}$, where we used estimate (\ref{esti4}),
 $\norm{\pr-\xi}\le C\norm{\pr-\xi}_\infty$ and $ \norm{\pr-\xi}_\infty \le 1$.
Plugging the above estimate into (\ref{lip1}) it follows $
|V(t,\pr)-V(t,\xi)| \le C\norm{\pr-\xi}$, which proves the Lipschitz
continuity of $V(t,\pr)$ in $\pr$.

\paragraph{Continuity in $t$}
Let $0\leq t<s\leq T$, then the dynamic programming principle to
$V(t,\pr)$ implies
\begin{eqnarray*}
0&\leq& |V(t,\pr)-V(s,\pr)|\\
&=&\sup_{h\in
\mathcal{H}}E\Big(\exp\Big\{-\int_t^{s}b(R\pr^{(t,\pr,h)}_u,h_u)du\Big\}\,V(s,\pr_{s}^{(t,\pr,h)})-V(s,\pr)\Big)\\
&\leq&\sup_{h\in
\mathcal{H}}E\Big(\exp\Big\{-\int_t^{s}b(R\pr^{(t,\pr,h)}_u,h_u)du\Big\}\,\Big|V(s,\pr_{s}^{(t,\pr,h)})-V(s,\pr)\Big|\Big)\\
&+&\sup_{h\in
\mathcal{H}}E\Big(\Big|\exp\Big\{-\int_t^{s}b(R\pr^{(t,\pr,h)}_u,h_u)du\Big\}\,V(s,\pr)-V(s,\pr)\Big|\Big).
\end{eqnarray*}
Using the Lipschitz continuity of $V$ in $\pr$ the first term can be
estimated by
\[C \sup_{h\in \mathcal{H}}E\Big(|\pr_{s}^{(t,\pr,h)}-\pr|\Big) \le C|s-t|^{\frac{1}{2}}\]
where we have used  \eqref{esti3}. For the second term the
boundedness of $b$ and $V$ yields the estimate
\[|e^{C_b(s-t)}-1|V(s,\pr) \le C|s-t|\]
where we have used that $f(x)=e^x  $ is  Lipschitz continuous on
bounded intervals. Finally, we obtain
\begin{eqnarray*}
|V(t,\pr)-V(s,\pr)|& \leq&C(|s-t|^{\frac{1}{2}}+|s-t|) \le
(C+T^{\frac{1}{2}})|s-t|^{\frac{1}{2}}.
\end{eqnarray*}
\end{proof}

\end{document}